\documentclass[11pt,a4paper]{article}
\usepackage{amsthm}
\usepackage{amssymb}
\usepackage{amsmath}
\usepackage{amsfonts}

\usepackage{epsfig}

\usepackage{graphics,graphicx}

\usepackage[ruled, lined]{algorithm2e}

\newtheorem{proposition}[equation]{Proposition}
\newtheorem{examp}[equation]{Example}

\theoremstyle{definition}
\newtheorem{definition}[equation]{Definition}
\newtheorem{remark}[equation]{Remark}

\numberwithin{equation}{section}

\def\nref#1{{\rm (\ref{#1})}}

\def\nref#1{{\rm (\ref{#1})}}

\begin{document}

\title{Comparing series of rankings with ties by using complex networks: An analysis of the spanish stock market (IBEX-35 index)}

\author{
F.\,Pedroche$^1$, R.\,Criado$^{2,3}$, E.\,Garc\'{\i}a$^{2,3}$, M.\,Romance$^{2,3}$ and V.E.\,S\'anchez$^4$\\
\\
{\small $^1$\,Institut de Matem\`{a}tica Multidisciplin\`{a}ria} \\
{\small Universitat Polit\`{e}cnica de Val\`{e}ncia} \\
{\small Cam\'{\i} de Vera s/n. 46022 Val\`{e}ncia, Spain.}\\
{\small pedroche@mat.upv.es}\\
\\
{\small $^2$\,Departamento de Matem\'{a}tica Aplicada} \\
{\small ESCET - Universidad Rey Juan Carlos} \\
{\small C/ Tulip\'an s/n 28933 M\'ostoles (Madrid), Spain.}\\
{\small regino.criado@urjc.es, esther.garcia@urjc.es, miguel.romance@urjc.es}\\
\\
{\small $^3$\,Centro de Tecnolog\'{\i}a Biom\'edica} \\
{\small Universidad Polit\'{e}cnica de Madrid} \\
{\small 28223 Pozuelo de Alarc\'on (Madrid), Spain.}\\
\\
{\small $^4$\,Departamento de Econom\'{\i}a de la Empresa} \\
{\small Universidad Carlos III} \\
{\small 28903 Getafe (Madrid) , Spain.}\\
}

\date{}

\maketitle

\begin{abstract}
In this paper we extend the concept of Competitivity Graph to compare series of rankings with ties ({\em partial rankings}).  We extend the usual method used to compute Kendall's coefficient for two partial rankings to the concept of evolutive Kendall's coefficient for a series of partial rankings. The theoretical framework consists of a four-layer multiplex network. Regarding the treatment of ties, our approach allows to define a tie between two values when they are close {\em enough}, depending on a threshold. We show an application using data from the Spanish Stock Market; we analyse the series of rankings defined by $25$ companies that have contributed to the IBEX-35 return and volatility values over the period 2003 to 2013.
\end{abstract}


\section{Introduction}\label{sec:intro}

The analysis of series of rankings has a vast tradition in the literature (see, e.g.~\cite{CoScSi,StKoPi,LaMe2012} for a review on applications) and in recent years has attracted the interest of big companies  that handle massive amount of data on the Internet. One of the most studied issues is the computation of a consensus (or aggregated) ranking that best summarizes a series of rankings. A typical application of that approach is the set of web pages shown following a user query on a web searcher \cite{Dw}. The studies interested in obtaining a consensus ranking  often define a distance (see~\cite{CoKrSe,Co06}) to properly define their goal: that is, a ranking that minimizes such a distance. One of the seminal works dealing with distances in rankings can be found in chapter II of~\cite{KeSn}. The problem of finding a ranking that minimizes the distance to a series of rankings, sometimes called Kemeny optimal aggregation problem~\cite{Bie}, is NP-hard (for $4$ rankings or more) but some approximation techniques exist, see e.g.~\cite{Ai,CoDaKa,Kar,BeBrNi}.

One of the most successful approaches to measure the distance between two rankings is based upon the number of permutations, or crossings, that occur when passing from one ranking to the other. The study of permutations in series of rankings has  a long history and its origin can be situated around the 40's in the works of Kendall~\cite{Kendall38,KeBa39}. The Kendall's $\tau$ correlation coefficient, a measure of the number of pairs of elements that flip their positions, has been extensively used. For example, in the 70's,  Kendall's coefficient (including the treatment of ranking with ties) was used to compare political alliances among states; see references in~\cite{Sig}.  Kendall's $\tau$ can be computed in $O(\, n \, \log \, n)$ time, being $n$ the number of elements to rank~\cite{Dw}.

When applied to web search, it is usually sufficient to worry about the top-k elements~\cite{IlBe}. As a consequence, the rest of the $n-k$ elements can be supposed to be tied. The treatment of ties when comparing rankings has an intrinsic importance and it is one of our main focus in this paper. Sometimes in the literature, for example in~\cite{Fa06}, rankings with ties are called {\sl partial rankings}. This kind of rankings (rankings with ties) arises naturally in many applications: for example, when users $(r_1,r_2,\ldots,r_m)$ score the service offered by some hotels ($1,2, \cdots, n$) on a scale of 1 to 5 there can be a user that assign the same score to different hotels, producing a tie in that ranking.

Some variations of Kendall's $\tau$, including treatment of ties, have been proposed. For example,  in~\cite{Cart} the authors focus on some problems arisen on Information Retrieval applications; in~\cite{EmMa} the authors analyse the treatment of ties and show some problems in the original treatment of Kendall; in~\cite{KuVa} the authors show different alternatives to weight a crossing depending on the initial and/or final position of the swapped elements; in~\cite{Fa06} the authors show the equivalence of some metrics, and set up  an appropriate framework to extend Kendall's $\tau$ to the case of partial rankings, by defining a  {\em Kendall distance with penalty parameter $p$}.

In the field of graphs the topic of counting crossings is known as permutation graphs (see, e.g.,~\cite{PnLeEv,Kr}).  The theoretical definition of a graph that shows successive permutations can be situated, as far as we know, in the work of~\cite{GoRoUr}.

In this paper we are interested in the crossings that occur between successive rankings when these rankings show ties. As a starting point we use the method presented in \cite{Cri13} based on  measures of complex networks, such as mean strength and mean degree.  In this paper we extend this technique to partial rankings and we make use of multiplex networks \cite{SzLaTh}, \cite{GoGa12}, \cite{Go13}, to achieve that goal.

The structure of the paper is the following. In Section~\ref{sec:def} we give the basic definitions about scores and rankings with ties. In Section~\ref{sec:network} we motivate our approach based on multiplex networks and we extend the concept of {\em Kendall distance with penalty parameter $p$} of~\cite{Fa06} to the case of a series of partial rankings.  Finally, in Section~\ref{sec:economy} we show an application to analyse the competitiveness of IBEX index from 2003 to 2013.

\section{Scores, rankings with ties and basic definitions}\label{sec:def}

The main goal of this paper is analysing series of rankings with ties by using complex networks, and therefore we start this section giving the basic definitions about rankings (with and without ties), scores and competitivity graphs. Roughly speaking a {\sl ranking} (with or without ties) of a finite set of elements $\mathcal{N}=\{1,\cdots,n\}$ is an {\sl ordering} of the elements of $\mathcal{N}$ with or without ties between elements. From a mathematical point of view, this idea can be rewritten in the following way:

\begin{definition}
Given a finite set $\mathcal{N}=\{1,\cdots,n\}$ a {\sl ranking $r$ (without ties)} of $\mathcal{N}$ is a total order $\prec_r$ on $\mathcal{N}$, i.e. it is a binary relationship on $\mathcal{N}$ such that
\begin{itemize}
 \item[{\it (i)}]   $\prec_r$ is {\sl reflexive}, i.e., $i\prec_r i$ for every $i\in\mathcal{N}$,
 \item[{\it (ii)}]  $\prec_r$ is {\sl antisymmetric}, i.e. if $i\prec_r j$ and $j\prec_r i$, then $i=j$,
 \item[{\it (iii)}] $\prec_r$ is {\sl transitive}, i.e. if $i\prec_r j$ and $j\prec_r k$, then $i\prec_r k$,
 \item[{\it (iv)}]  $\prec_r$ is {\sl total}, i.e. if $i\ne j\in\mathcal{N}$, then either $i\prec_r j$ or $j\prec_r i$.
\end{itemize}
\end{definition}

\begin{remark}\label{rem:biyeccion}
In other references of the literature, a {\sl ranking $r$ (without ties)} of $\mathcal{N}$ is defined as a bijection $r:\mathcal{N}\longrightarrow \mathcal{N}$, since any bijection $r$ defines a  binary relationship $\prec_r$ on $\mathcal{N}$ given by $i\prec_r j$ if and only if $r(i)>r(j)$ that it is straightforward to check that it is a total order on $\mathcal{N}$. Similarly, it can be proved that if $\prec_r$ is a total order on $\mathcal{N}$, a bijection $\sigma:\mathcal{N}\longrightarrow \mathcal{N}$ can be defined such that $\sigma(i)\le \sigma(j)$ if and only if $j\prec_r i$. Furthermore, a {\sl ranking $r$ (without ties)} of $\mathcal{N}$ can be defined as any injective function $r:\mathcal{N}\longrightarrow \mathbb{R}$, as we will point out later when we remind the concept of score.
\end{remark}

Note that if we take a ranking without ties, if we consider $i\ne j\in\mathcal{N}$, then either $i\prec_r j$ or $j\prec_r i$, and therefore , if $i\ne j$, they are always ranked differently. This fact is not always possible in real cases, since, for example, two teams can be tied in a competition or two web pages could have the same PageRank and therefore we should also consider rankings with ties, as the following definition shows.

\begin{definition}
Given a finite set $\mathcal{N}=\{1,\cdots,n\}$ a {\sl ranking $r$ with ties} of $\mathcal{N}$ is a weak order $\prec_r$ on $\mathcal{N}$, i.e. it is a binary relationship on $\mathcal{N}$ such that
\begin{itemize}
 \item[{\it (i)}]   $\prec_r$ is {\sl reflexive}, i.e., $i\prec_r i$ for every $i\in\mathcal{N}$,
 \item[{\it (ii)}] $\prec_r$ is {\sl transitive}, i.e. if $i\prec_r j$ and $j\prec_r k$, then $i\prec_r k$,
 \item[{\it (iii)}]  $\prec_r$ is {\sl total}, i.e. if $i\ne j\in\mathcal{N}$, then either $i\prec_r j$ or $j\prec_r i$.
\end{itemize}
If $r$ is a ranking with ties, and we have two different elements $i\ne j\in\mathcal{N}$, we say that $i$ and $j$ {\sl are tied} if $i\prec_r j$ and $j\prec_r i$.
\end{definition}

Many times in real problems, the rankings (with or without ties) of a finite set $\mathcal{N}=\{1,\cdots,n\}$ are defined from a numerical function that {\sl weights} the relevance of each element of $\mathcal{N}$: a {\sl score}.

\begin{definition}
Given a finite set $\mathcal{N}=\{1,\cdots,n\}$ a {\sl score} (also called a {\sl rating}) of $\mathcal{N}$ is a function $s:\mathcal{N}\longrightarrow\mathbb{R}$.
\end{definition}

If we take a score $s:\mathcal{N}\longrightarrow\mathbb{R}$, it is straightforward to check that $s$ induces a ranking (with or without ties) $r=r_s$ given for every pair of elements $i,j\in \mathcal{N}$ by
\[
i\prec_r j\Longleftrightarrow s(i)\ge s(j).
\]
It is easy to check that several scores $s$ can produce the same ranking $r_s$. Roughly speaking, the ranking is the {\sl qualitative} version of an ordering while the corresponding score is the {\sl quantitative} version of it, but in many real-life applications it is enough considering the qualitative version, i.e. the ranking itself.

Note that $r$ is a ranking without ties if and only if the score $s$ is injective, as we pointed out in Remark~\ref{rem:biyeccion} and since the fact that two different elements $i\ne j\in\mathcal{N}$ are tied (with respect to $r_s$) if and only if $s(i)=s(j)$. This fact makes that the existence of ties in many real-life rankings can be very unstable: if the raking is given by an empirical score with non-integer values, the ties could disappear by round-off errors. In order to avoid this problem, we suggest considering {\sl approximated ties}, when the scores of two elements are below a (small) precision threshold, as it is introduced in the following definition.

\begin{definition}\label{def:aprox}
Given a finite set $\mathcal{N}=\{1,\cdots,n\}$, a score $s:\mathcal{N}\longrightarrow\mathbb{R}$ and $\Delta x\in(0,+\infty)$, we define the {\sl precision intervals} $[L_1,L_2)$, $[L_2,L_3)$,$\cdots$, $[L_{\hat{n}},L_{\hat{n}+1})$, where:
\begin{itemize}
 \item[{\it (i)}] $\displaystyle \hat{n}=\left\lfloor \frac {M-m}{\Delta x}+1\right\rfloor$, with $M=\max\{s(i);\enspace i\in\mathcal{N}\}$, $m=\min\{s(i);\enspace i\in\mathcal{N}\}$ and $\lfloor\cdot\rfloor$ is  the floor function,
 \item[{\it (ii)}] $\displaystyle \delta x=\frac {M-m+\Delta x}{\hat{n}}$,
 \item[{\it (iii)}] $\displaystyle L_1=m-\frac{\Delta x}2$,
 \item[{\it (iv)}] $L_i=L_1+(i-1)\delta x$, for every $2\le i\le \hat{n}+1$.
\end{itemize}
The {\sl ranking with (approximated) ties} associated to the score $s$ and precision threshold  $\Delta x$ is the ranking $\prec$ (also denoted by $\prec_{s,\Delta x})$ such that if $i\ne j\in\mathcal{N}$
\begin{itemize}
 \item[{\it (i)}] $i\prec j$ but $j\nprec i$ if and only if  there is $1\le k<\ell\le \hat{n}$ such that $s(j)\in [L_k,L_{k+1})$ and $s(i)\in [L_\ell,L_{\ell+1})$,
 \item[{\it (ii)}] $i$ and $j$ are tied if there is $1\le k\le \hat{n}$ such that $s(i),s(j)\in [L_k,L_{k+1})$.
\end{itemize}
\end{definition}

The intuitive idea behind the formalism of the previous definition is considering ties between elements whose scores are below a fixed precision threshold, as it is presented in the following example. In the sequel, we always construct ranking with (approximated) ties associated to a score $s$ and a precision threshold  $\Delta x$ and we simply call it simply {\sl ranking with ties associated to the score $s$ and with precision threshold  $\Delta x$}.

\begin{examp}\label{ex:5scores}
Let us consider the following $m=5$ scores of elements $\mathcal{N}=\{1,2,3,4,5,6\}$ given by the columns of Table~\ref{ta:scores}.
\begin{table}[h]
\begin{center}
\[
\begin{array}{|c|c|c|c|c|c|}
\hline
& s_1  & s_2 & s_3 & s_4 & s_5 \\
\hline
1& 3  &  5 & -7 & -7 & 20 \\
2& 23 &  7 & 24 & 20 &  8 \\
3& 22 &  7 & 20 & 10 &  8 \\
4& -5 &  3 & -7 & -5 &  9 \\
5& 22 & 15 & 30 & 20 &  8 \\
6&  3 & 12 & 30 & 10 & 15 \\
\hline
\end{array}
\]
\caption{A family of five scores on the set $\mathcal{N}=\{1,\cdots,6\}$}
\label{ta:scores}
\end{center}
\end{table}

If we consider the precision threshold  $\Delta x=0.5$, then the rankings with (approximated) ties associated to the scores $s_1,\cdots, s_5$ is given in Table~\ref{ta:rankings}.
\begin{table}[h]
\begin{center}
\[
\begin{array}{|c|c|c|c|c|}
\hline
\sigma_1  & \sigma_2 & \sigma_3 & \sigma_4 & \sigma_5 \\
\hline
2   &  5  & 5,6 & 2,5   & 1 \\
3,5 & 6   & 2   & 3,6   & 6\\
1,6 & 2,3 & 3   & 4     & 4\\
4   & 1   & 1,4 & 1     & 2,3,5\\
    & 4   &     &       &  \\
\hline
\end{array}
\]
\caption{Rankings with (approximated) ties associated to the scores $s_1,\cdots, s_5$ and precision threshold  $\Delta x=0.5$}
\label{ta:rankings}
\end{center}
\end{table}

In ranking $\sigma_1$ the first position is for node $2$, while in the second position we have a tie between the nodes $3$ and $5$; in third position we have a tie between nodes $1$ and $6$. The last position is for node $4$.

On the other hand, if we take the precision threshold  $\Delta x=2$, then we obtain the rankings with (approximated) ties associated to the scores $s_1,\cdots, s_5$ presented in Table~\ref{ta:rankings2}.
\begin{table}[h]
\begin{center}
\[
\begin{array}{|c|c|c|c|c|}
\hline
\sigma_1  & \sigma_2 & \sigma_3 & \sigma_4 & \sigma_5 \\
\hline
2,3,5   &  5  & 5,6 & 2,5   & 1 \\
1,6 & 6   & 2   & 3,6   & 6\\
4   & 2,3 & 3   & 4     & 4\\
    & 1   & 1,4 & 1     & 2,3,5\\
    & 4   &     &       &  \\
\hline
\end{array}
\]
\caption{Rankings with (approximated) ties associated to the scores $s_1,\cdots, s_5$ and precision threshold  $\Delta x=2$}
\label{ta:rankings2}
\end{center}
\end{table}

In this case, the ranking $\sigma_1$, in the first position we have a triple tie among nodes $2$, $3$ and $5$ and in second position there is a tie between nodes $1$ and $6$.
\end{examp}

Once we have presented all the notation needed about rankings, ties, scores and precision thresholds, we introduce the second ingredient of this paper: the {\sl competitivity networks} and the {\sl evolutive competitivity networks}, introduced in \cite{Cri13}.

\begin{definition}
If we take a family of rankings (without ties) ${\mathcal R}=\{\sigma_1, \dots, \sigma_m \}$ of elements ${\mathcal N}=\{1,\dots, n\}$,  we say that the pair of elements $(i,j)\in {\mathcal N}$ {\sl compete} if there exist $c_s,c_t\in\{1,2,\ldots, m\}$ such that $i\prec_{\sigma_s} j$ but $j\prec_{\sigma_t} i$, i.e.,   $i$ and $j$ exchange their relative positions between the rankings $\sigma_{s}$ and $\sigma_{t}$.

We define the {\it competitivity network} of the family of rankings ${\mathcal R}$ as the undirected network  $G_c({\mathcal R})=({\mathcal N}, E_{\mathcal R})$, where the set of edges $E_{\mathcal R}$ is given by the following rule: there is a link between nodes $i$ and $j$ if $(i,j)$ compete. This notion had already been introduced in \cite{GoRoUr} as intersection graphs of concatenations of permutation diagrams, and was shown to being equivalent to the notions of co-comparability graphs and to $f$-graphs \cite[Theorem 1]{GoRoUr}.

We say that the pair of elements $(i,j)\in {\mathcal N}$ {\sl compete $k$-times}  if $k$ is the maximal number of rankings where $i$ and $j$ compete. The {\it evolutive competitivity network} of ${\mathcal R}$, denoted by $G_c^e({\mathcal R})=({\mathcal N}, E_{\mathcal R}^e)$, will be the weighted undirected network, where the set of edges $E_{\mathcal R}^e$ is given by the rule: there is a link between nodes $i$ and $j$ labeled with weight $k$ if $(i,j)$ compete $k$ times. Note that the underlying (unweighted) network behind the (weighted) network $G_c^e({\mathcal R})$ is $G_c({\mathcal R})$.
\end{definition}

There are several ways of extending this notion to families of rankings with ties, but in the next section, we will see how to define the {\sl multiplex (evolutive) competitivity network}, which is the natural extension of competitivity graphs that helps to analyse properties of families of rankings with ties.

\section{Multiplex networks associated to a family of rankings with ties}\label{sec:network}

In this section we introduce a multiplex network associated to a  family of rankings, some of them with ties among its nodes.  We extend the notion of Kendall distance with penalty and show that can also be computed by considering the associated projected multiplex network. Moreover, we relate the classical notion of normalized mean strength of such network with the extension of the Kendall distance.

\subsection{The multiplex evolutive competitivity network associated to a family of rankings with ties}

In this subsection we introduce the notion of  multiplex evolutive competitivity network associated to a family of rankings with ties, extending the previous notion of \cite{Cri13} defined for rankings with no ties.

\begin{definition}
Let $\alpha, p, q, \gamma\ge 0$ be four given parameters. Given a set of $n$ nodes ${\mathcal N}=\{1,\dots, n\}$ and a  finite family of rankings with ties ${\mathcal R}=\{\sigma_1,\dots, \sigma_m \}$ of ${\mathcal N}$, we  define  {\it the multiplex evolutive competitivity network of ${\mathcal R}$} in the following way: this multiplex network, denoted by $\mathcal{M}G(\sigma_1,\cdots,\sigma_m)$, contains four layers, called {\it the crossing layer}, {\it the semi-crossing layer}, {\it the long-term-crossing layer} and the {\it tie layer}. All of them have $n$ nodes  $\{1,\dots, n\}$ and the edges in each layer are given as follows:
\begin{itemize}
\item[{-}] The edges in the crossing layer are defined in the following way: there is an edge between nodes $i$ and $j$ labelled with weight $k\alpha$ if  $i$ and $j$ exchange their relative positions  between two consecutive rankings $\sigma_{t_\ell}$ and $\sigma_{t_{\ell}+1}$, $\ell=1,\dots,k$, in ${\mathcal R}$.
\item[{-}]  The edges in the semi-crossing layer   are defined in the following way: there is an edge between nodes $i$ and $j$ labelled with weight $kp$  if  there exist $k$ consecutive rankings $\sigma_{t_\ell}$ and $\sigma_{t_{\ell}+1}$, $\ell=1,\dots,k$, in ${\mathcal R}$ such that $i$ and $j$ are tied in $\sigma_{t_\ell}$ and not tied in $\sigma_{t_{\ell}+1}$ or $i$ and $j$ are not tied in $\sigma_{t_\ell}$ but tied in $\sigma_{t_{\ell}+1}$, i.e., nodes $i$ and $j$ go from tied to not tied or not tied to tied in $k$ consecutive rankings of $\mathcal{R}$.
\item[{-}] The edges in the long-term-crossing layer  are defined in the following way: there is an edge between nodes $i$ and $j$ labelled with weight $kq$ if there exists a maximal set of rankings $\sigma_{t_1},\dots, \sigma_{t_k}\in {\mathcal R}$ such that for each $\ell=1,\dots, k$ the pair $i$ and $j$ are not tied in $\sigma_{t_\ell}$, are tied in $\sigma_{t_{\ell}+1},\sigma_{t_\ell+2}, \ldots, \sigma_{t_\ell+s}$, with $s \geq 1$, are not tied in $\sigma_{t_\ell+s+1}$ and $i$ and $j$ exchange their relative positions between $\sigma_{t_\ell}$ and  $\sigma_{t_\ell+s+1}$.
\item[{-}] The edges in the tie layer are defined in the following way: there is an edge between nodes $i$ and $j$ labelled with weight $k\gamma $ if  there exist $k$ consecutive rankings $\sigma_{t_\ell}$ and $\sigma_{t_{\ell}+1}$, $\ell=1,\dots,k$, in ${\mathcal R}$ such that $i$ and $j$ are both tied in $\sigma_{t_\ell}$ and  $\sigma_{t_{\ell}+1}$, i.e., nodes $i$ and $j$ are tied in $k$ pairs of consecutive rankings of $\mathcal{R}$.
\end{itemize}
\end{definition}

The use of multiplex networks in the context of competitivity graphs allows to analyse the interplay between ties, crossings and long term crossings since, multiplex networks are a sharp tool for studying complex systems with heterogeneous interactions \cite{NT}.

\begin{definition}\label{projev}
The {\it projected evolutive competitivity network} or simply the {\it evolutive competitivity network } associated to  a  finite family of rankings with ties ${\mathcal R}=\{\sigma_1,\dots, \sigma_m \}$ of $n$ nodes ${\mathcal N}=\{1,\dots, n\}$ is just the projection of the four layers of the multiplex evolutive competitivity network onto a single label, i.e., consists on a network of $n$ nodes, and edges between pairs of nodes $i$ and $j$ with weight $k_1\alpha+k_2p+k_3 q+k_4\gamma$ if there is a link between $i$ and $j$ of weight $k_1\alpha$ in the crossing layer, a link  between $i$ and $j$ of weight $k_2p$ in the semi-crossing layer,  a link between $i$ and $j$ of weight $k_3 q$ in the long-term-crossing layer, and a link of weight $k_4 \gamma$ in the tie layer. We denote it by $\mathcal{P}G(\sigma_1,\cdots,\sigma_m)$
\end{definition}

If none of the rankings of the family  ${\mathcal R}=\{\sigma_1,\dots, \sigma_m \}$ contains any tie, then there are no weighted edges in the semi-crossing label, no weighted edges in the long-term-crossing layer  and  no weighted edges in the tie label of the multiplex evolutive competitivity network of ${\mathcal R}$, so the evolutive competitivity graph defined in \cite{Cri13} corresponds to the projected evolutive competitivity network associated to  ${\mathcal R}$ with fixed $\alpha=1$.

\begin{examp}\label{exa:multiplex}
Let us consider the following family of rankings with ties ${\mathcal R}=\{\sigma_1, \sigma_2, \sigma_3, \sigma_4, \sigma_5\}$ of 6 nodes ${\mathcal N}=\{1, 2, 3, 4, 5, 6\}$, obtained in Example~\ref{ex:5scores}.
\[
\begin{array}{|c|c|c|c|c|}
\hline
\sigma_1  & \sigma_2 & \sigma_3 & \sigma_4 & \sigma_5 \\
\hline
2   &  5  & 5,6 & 2,5   & 1 \\
3,5 & 6   & 2   & 3,6   & 6\\
1,6 & 2,3 & 3   & 4     & 4\\
4   & 1   & 1,4 & 1     & 2,3,5\\
    & 4   &     &       &  \\
\hline
\end{array}
\]

These rankings were obtained with a precision threshold $\Delta=0.5$ from the family of scores $\{s_1,\cdots,s_5\}$ (see Table~\ref{ta:rankings}). The multiplex evolutive competitivity network $\mathcal{M}G(\sigma_1,\cdots,\sigma_5)$ associated to ${\mathcal R}$ and the projected evolutive competitivity network $\mathcal{P}G(\sigma_1,\cdots,\sigma_5)$ associated to ${\mathcal R}$ are presented in Figure~\ref{fig:multiplex}. Note that in this case the number of crossings (appearing in layer $\ell_1$) and the number of semi-crossings  (appearing in layer $\ell_2$) is much bigger than the number of long-term-crossings (layer $\ell_3$) and ties (layer $\ell_4$).


\begin{figure}[h!]
\begin{center}
 \includegraphics[width=\textwidth]{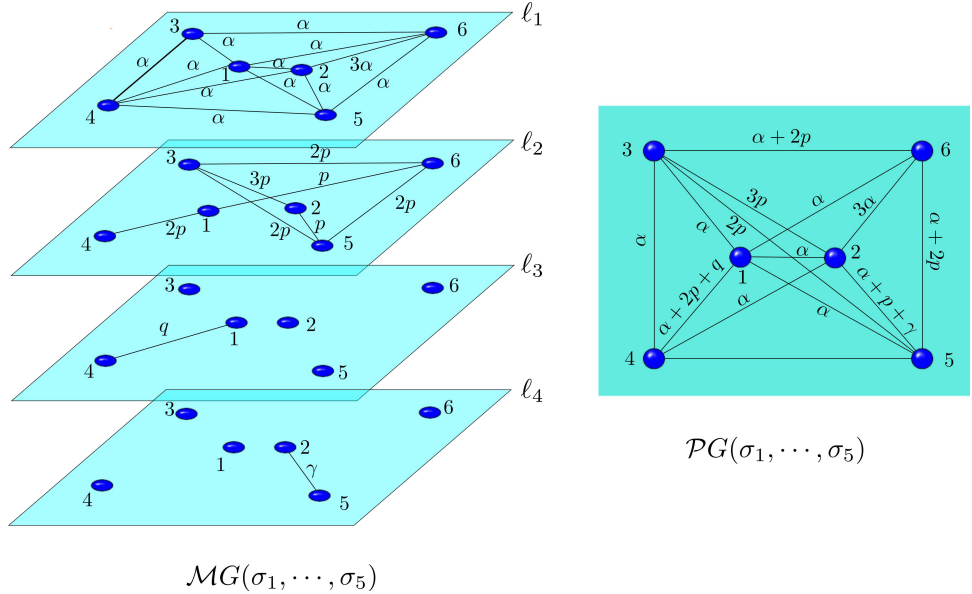}
\end{center}
\caption{Multiplex evolutive competitivity network $\mathcal{M}G(\sigma_1,\cdots,\sigma_5)$ associated to the family of rankings ${\mathcal R}$ and the projected evolutive competitivity network $\mathcal{P}G(\sigma_1,\cdots,\sigma_5)$ associated to ${\mathcal R}$. Layer $\ell_1$ is the crossing layer, $\ell_2$ is the semi-crossing layer, $\ell_3$ is the long-term-crossing layer and $\ell_4$ is the tie layer.}
\label{fig:multiplex}
\end{figure}

\end{examp}

\subsection{The evolutive Kendall distance}

Now we extend the notion of Kendall distance given in \cite{Fa06} for two rankings with ties. Let us first recall the original definition for two rankings with ties:

\begin{definition}
Let $\sigma_1$ and $\sigma_2$ be two rankings with ties of the nodes ${\mathcal N}=\{1,\dots, n\}$, $\alpha\in[0,+\infty)$ and  a penalty $p \in [0,\frac 12]$. The  {\it Kendall distance with penalty parameter $p$} is defined as
\begin{equation}\label{Kdp}
K^{(p)}(\sigma_1,\sigma_2) = \sum_{ \{i,j\} \in {\mathcal{N}} } \bar{K}_{i,j}^{(p)}(\sigma_1,\sigma_2)
\end{equation}
where  $\bar{K}_{i,j}^{(p)}(\sigma_1,\sigma_2)$ corresponds to one of the three following cases:
\begin{enumerate}
\item[{Case 1.}] If $i$ and $j$ are not tied in $\sigma_1$ nor in $\sigma_2$ and they exchange their relative positions between $\sigma_1$ and $\sigma_2$ then $\bar{K}_{i,j}^{(p)}=\alpha$. Otherwise $\bar{K}_{i,j}^{(p)}=0$.
\item[{Case 2.}] If  $i$ and $j$ are tied in both $\sigma_1$ and $\sigma_2$ then
 $\bar{K}_{i,j}^{(p)}=0$.
\item[{Case 3.}] If $i$ and $j$ are tied in one of $\sigma_1$ or $\sigma_2$
but not tied in the other  ranking then $\bar{K}_{i,j}^{(p)}=p$.
\end{enumerate}

Notice that Kendall distance with penalty parameter $p$ of \cite{Fa06} is exactly the addition of the weights of all the edges of the evolutive competitivity network given in Definition \ref{projev} for the family of rankings ${\mathcal R}=\{\sigma_1, \sigma_2\}$ with fixed parameters $\alpha=1$,  $p \in [0,\frac 12]$,  and $\gamma=0$.
\end{definition}

The extension of the previous concept to a Kendall distance for a family of $m$ ordered  rankings with ties
${\mathcal R}=\{ \sigma_1, \sigma_2, \ldots, \sigma_m \}$ is straightforward:

\begin{definition}\label{Kevd}
 Given a set of $n$ nodes ${\mathcal N}=\{1,\dots, n\}$,  $\alpha\in[0,+\infty)$,  $p \in [0,\frac 12]$ and a  finite family of rankings with ties ${\mathcal R}=\{\sigma_1,\dots, \sigma_m \}$ of ${\mathcal N}$, we can define {\it the evolutive Kendall distance with penalty parameter $p$} as
\begin{equation}\label{Kevd2}
K^{(p)}_{ev}(\sigma_1,\sigma_2, \ldots, \sigma_m) = \sum_{i=1}^{m-1} K^{(p)}(\sigma_i,\sigma_{i+1}).
\end{equation}
\end{definition}

\begin{examp}\label{exa:1}
Let us consider again the family of rankings with ties given in Example~\ref{exa:multiplex} and fix $\alpha=2$. In order to obtain the evolutive Kendall's distance given by \nref{Kevd2} we need to compute the $4$ Kendall's distances $K^{(p)}(\sigma_i,\sigma_{i+1})$. The first one is:
\[
K^{(p)}(\sigma_1,\sigma_2) = \sum_{ \{i,j\} \in {\mathcal{N}} } \bar{K}_{i,j}^{(p)}(\sigma_1,\sigma_2)
\]
where we have ${\mathcal{N}}=\{1,2,\ldots, 6\}$ and  there are $15$ pairs of indices $\{i,j\}$. The pairs that contribute with non-zero penalties are shown in Table~\ref{ta:ex1}.
\begin{table}[h]
\begin{center}
\begin{tabular}{|c|c|l|}
\hline
$\{i,j\}$ &  $\bar{K}_{i,j}^{(p)}(\sigma_1,\sigma_2)$ &  \\
\hline
$1,6$  & $p$ & since they go from tied to untied from $\sigma_1$ to $\sigma_2$ \\
\hline
$2,3$  & $p$ & since they go from tied to untied from $\sigma_1$ to $\sigma_2$ \\
\hline
$2,5$  &  $2$ & since they cross from $\sigma_1$ to $\sigma_2$ \\
\hline
$2,6$  &  $2$ & since they cross from $\sigma_1$ to $\sigma_2$ \\
\hline
$3,5$  & $p$ & since they go from tied to untied from $\sigma_1$ to $\sigma_2$ \\
\hline
$3,6$  & $2$ & since they cross from $\sigma_1$ to $\sigma_2$ \\
\hline
\end{tabular}
\caption{Details of the computation of $K^{(p)}(\sigma_1,\sigma_2)$}
\label{ta:ex1}
\end{center}
\end{table}
Therefore,
\[
K^{(p)}(\sigma_1,\sigma_2) = 6 + 3p.
\]
In an analogous way it is easy to compute the rest of the Kendall distances $ K^{(p)}(\sigma_i,\sigma_{i+1})$
\[
K^{(p)}(\sigma_2,\sigma_3)=3p, \quad K^{(p)}(\sigma_3,\sigma_4)=2+4p, \quad K^{(p)}(\sigma_4,\sigma_5)=20+3p.
\]
Therefore, by Definition~\ref{Kevd} we obtain
\[
K^{(p)}_{ev} = 28 + 13p.
\]

Notice that this value of $K^{(p)}_{ev}$ could also be obtained from the projected evolutive competitivity network of $\{\sigma_1, \sigma_2, \sigma_3, \sigma_4, \sigma_5\}$ for  $\alpha=2$, $\gamma=0$  (see Figure~\ref{fig:multiplex}) by computing the sum of  the weights of all the edges of this network.
\end{examp}

In  the previous example, we would like to point out the following fact: The evolutions of the positions of the pair $\{ 1,4\}$ from non adjacent rankings $\sigma_2$ to $\sigma_4$ and the pair $\{3,6\}$ from non adjacent rankings $\sigma_3$ to $\sigma_5$ are different. There has been a crossing between $\{ 1,4\}$ that has not been counted and there has been no crossings between $3$ and $6$. Nevertheless, these two cases  contribute with a penalty of $2p$ to the calculation of the  evolutive Kendall distance.

In order to take into account this long term crossings we introduce a new case in  the list of penalties for the case of evolutive Kendall's distance for $m$ rankings with ties.

\begin{definition} \label{def:Kevdc}
Given a family ${\mathcal R}=\{ \sigma_1, \sigma_2, \ldots, \sigma_m\}$ of $m$ rankings with ties of nodes ${\mathcal N}=\{1,\dots, n\}$,  $\alpha\in[0,+\infty)$ and  $p \in [0,\frac 12]$, we define  the {\em corrected evolutive Kendall distance with penalty parameter $p$} as follows:
\begin{equation}\label{Kevdc}
K^{(p)}_{cev}(\sigma_1, \ldots, \sigma_m) = K^{(p)}_{ev}(\sigma_1, \ldots, \sigma_m)
+ \sum_{ \{i,j\} \, \text{verifies Case 4}} \bar{K}_{i,j}^{c}(\sigma_1, \ldots, \sigma_m),
\end{equation}
where $\bar{K}_{i,j}^{c}(\sigma_1,\sigma_2, \ldots, \sigma_m)$ is a penalty for the pairs $\{ i,j\}$ that verify the following case:
\begin{enumerate}
\item[Case 4.]   if there exists a maximal set of rankings $\sigma_{t_1},\dots, \sigma_{t_k}\in {\mathcal R}$ such that for each $\ell=1,\dots, k$ the pair $i$ and $j$ are not tied in $\sigma_{t_\ell}$, are tied in $\sigma_{t_{\ell}+1},\sigma_{t_\ell+2}, \ldots, \sigma_{t_\ell+s}$, with $s \geq 1$, are not tied in $\sigma_{t_\ell+s+1}$ and $i$ an $j$ exchange their relative positions between $\sigma_{t_\ell}$ and  $\sigma_{t_\ell+s+1}$. In this case $\bar{K}_{i,j}^{c}(\sigma_1,\sigma_2, \ldots, \sigma_m)=k$, where $k$ is the number of rankings in the maximal  set of rankings $\sigma_{t_1},\dots, \sigma_{t_k}\in {\mathcal R}$ verifying the previous property.
\end{enumerate}
\end{definition}

\begin{remark}\label{rem:ns-cev}
Note that it is easy to check that previous definition coincides with the addition of the weights of all the edges of the evolutive competitivity network given in Definition~\ref{projev} for the family of rankings ${\mathcal R}=\{\sigma_1, \dots,  \sigma_m\}$ with fixed parameters $\alpha=2$,  $p \in [0,\frac 12]$, $q=1$  and $\gamma=0$.
\end{remark}

\begin{examp}\label{exa:2}
Let us consider the $5$ rankings of Example~\ref{exa:1},  $\alpha=2$ and  a penalty $p \in [0,\frac 12]$. It is easy to see that the only pair of nodes that verifies Case 4 is the pair $\{1,4\}$ commented  above. Therefore
\[
\sum_{ \{i,j\} \, \text{verifies case 4}} \bar{K}_{i,j}^{c}(\sigma_1, \ldots, \sigma_5) =
\bar{K}^{c}_{1,4}(\sigma_1,\ldots,\sigma_{5}) = 1
\]
and
\[
\begin{split}
K^{(p)}_{cev}(\sigma_1, \ldots, \sigma_5) &=
K^{(p)}_{ev}(\sigma_1, \ldots, \sigma_5) +  \sum_{ \{i,j\} \, \text{verifies case 4}} \bar{K}_{i,j}^{c}(\sigma_1, \ldots, \sigma_5) \\
 &= 28+13p+1=29 + 13p.
\end{split}
\]
\end{examp}

Note that the corrected evolutive Kendall distance the evolutive Kendall distance are strongly related, but, as we will see in Section~\ref{sec:economy} they can give quite different values in some cases. The next result gives the analytical relationship between $K^{(p)}_{cev}(\sigma_1, \ldots, \sigma_m)$ and $K^{(p)}_{ev}(\sigma_1, \ldots, \sigma_m)$.

\begin{proposition}\label{prop:estimate}
Given  a family ${\mathcal R}=\{ \sigma_1, \sigma_2, \ldots, \sigma_m\}$ of $m$ rankings with ties of nodes ${\mathcal N}=\{1,\dots, n\}$ and $\alpha,p\in \mathbb{R}$, then
\begin{equation}\label{eq:estimate}
K^{(p)}_{ev}(\sigma_1, \ldots, \sigma_m)\le K^{(p)}_{cev}(\sigma_1, \ldots, \sigma_m)\le K^{(p)}_{ev}(\sigma_1, \ldots, \sigma_m)+\frac{n(n-1)}2\left\lfloor\frac{m-1}2\right\rfloor,
\end{equation}
where $\lfloor\cdot\rfloor$ is the floor function. Furthermore, these two inequalities are sharp.
\end{proposition}

\begin{proof}
On the one hand, it is straightforward to check that
\[
K^{(p)}_{cev}(\sigma_1, \ldots, \sigma_m)\le K^{(p)}_{cev}(\sigma_1, \ldots, \sigma_m),
\]
simply by using \nref{Kevdc}. This first inequality is obviously sharp if and only if the family ${\mathcal R}$ of rankings has no long-term-crossings.

On the other hand,  in order to get the upper bound, it is enough to maximize the value of the expression
\begin{equation}\label{eq:case4}
\sum_{ \{i,j\} \, \text{verifies Case 4}} \bar{K}_{i,j}^{c}(\sigma_1, \ldots, \sigma_m).
\end{equation}
Note that the maximal value of \nref{eq:case4} could be derived from two principles: maximizing the number of pairs $i,j\in\mathcal{N}$ that verify the Case 4 in the definition of $K^{(p)}_{cev}(\sigma_1, \ldots, \sigma_m)$ and maximizing the value of $\bar{K}_{i,j}^{c}(\sigma_1, \ldots, \sigma_m)$ for each pair $i,j\in\mathcal{N}$ that verifies the Case 4. If we want to maximize the number of pairs of nodes that verifies the Case 4, we should make that all the pairs of nodes has a long-term-crossing. This can be performed, for example, if the family of rankings ${\mathcal R}$ contains the following rankings as non-consecutive elements of ${\mathcal R}$:
\[
\begin{array}{|c|c|}
\hline
\sigma  & \bar\sigma \\
\hline
1      &  n     \\
2      & n-1    \\
\cdots & \cdots \\
n-1    & 2      \\
 n     & 1       \\
\hline
\end{array}
\]
If $\sigma,\bar\sigma\in\mathcal{R}$ and they are not consecutive (i.e. there is $1\le i,j\le n$ with $|i-j|>1$ such that $\sigma_i=\sigma$ and $\sigma_j=\bar\sigma$), then all the pairs have a long-term-crossing, since all the pairs  exchange their relative position between $\sigma_i$ and $\sigma_j$ and in addition to this,  $\sigma_i$ and $\sigma_j$ are not consecutive rankings.

If we want to maximize the value of $\bar{K}_{i,j}^{c}(\sigma_1, \ldots, \sigma_m)$ for each pair $i,j\in\mathcal{N}$ that verifies the Case 4, we should consider a family of rankings such that maximizes the number of long-term-crossings between each pair of nodes. Since the minimal number of consecutive rankings in order to get a long-term-crossing is 3 (if we fix $i,j\in\mathcal{N}$, we need at least three rankings $\sigma_{k},\sigma_{k+1},\sigma_{k+2}$ such that $i$ and $j$ are tied in $\sigma_{k+1}$, they are not tied in $\sigma_k$ and $\sigma_{k+2}$ and they have different relative ordering in $\sigma_k$ and $\sigma_{k+2}$), then it is easy to check that the maximal number of long-term-crossing between a pair of nodes is $\left\lfloor\frac{m-1}2\right\rfloor$, where  $\lfloor\cdot\rfloor$ is the floor function.

Hence
\[
\begin{split}
\sum_{ \{i,j\} \, \text{verifies Case 4}} \!\!\!\!\!\!\!\!\!\!\!\!\bar{K}_{i,j}^{c}(\sigma_1, \ldots, \sigma_m)&\le \!\!\!\!\!\!\!\!\!\!\!\!\sum_{ \{i,j\} \, \text{verifies Case 4}} \!\!\!\!\!\!\!\!\!\!\!\!\max\{\bar{K}_{s,t}^{c}(\sigma_1, \ldots, \sigma_m);\enspace \{s,t\} \, \text{verifies Case 4}\}\\
&\le \sum_{ i\ne j\in\mathcal{N}} \max\{\bar{K}_{i,j}^{c}(\sigma_1, \ldots, \sigma_m);\enspace \{s,t\} \, \text{verifies Case 4}\}\\
&=\frac{n(n-1)}2\max\{\bar{K}_{i,j}^{c}(\sigma_1, \ldots, \sigma_m); \{s,t\}\, \text{verifies Case 4}\}\\
&\le \frac{n(n-1)}2\left\lfloor\frac{m-1}2\right\rfloor.
\end{split}
\]
Finally note that this bound is attained if we consider, for example the family of rankings ${\mathcal R}=\{ \sigma_1, \sigma_2, \ldots, \sigma_m\}$, where
\[
\sigma_i=\left\{
\begin{array}{cl}
\sigma& \text{if $i=3k-2$ for some $k\in\mathbb{N}$}\\
\sigma_o&\text{if $i=3k-1$ for some $k\in\mathbb{N}$}\\
\bar\sigma&\text{if $i=3k$ for some $k\in\mathbb{N}$}\\
\end{array}
\right.
\]
where $\sigma$ and $\bar\sigma$ are given before and $\sigma_o$ is the ranking with all the nodes tied. In this case all the pair of nodes verify Case 4 and the number of long-term-crossings is maximal, which makes that the upper bound is also attained.
\end{proof}

\subsection{The  Normalized Mean Strength}

The classical network notion of normalized mean strength of the evolutive competitivity network of a family of rankings with ties and the notion of corrected evolutive Kendall distance with penalty parameter $p$ of the same family of rankings are strongly related.

\begin{definition}
Given an undirected weighted network of $n$ nodes, the {\it strength $S(i)$ of each node $i$} is the sum of the weights of the edges incident to $i$. The {\it mean strength $MS$ of the network } is the sum of all the node strengths divided by the total number of nodes $n$, and {\it the normalized mean strength $NS$ of the network} is its mean strength  divided by the mean strength of a complete network of $n$ nodes with maximal weight in each edge.
\end{definition}

The connection between the normalized mean strength of the evolutive competitivity network of a family of rankings with ties and the corrected evolutive Kendall distance of the same family of rankings is given in the following result:

\begin{proposition}\label{prop_ns-cev}
Given a family of rankings with ties ${\mathcal R}=\{\sigma_1, \dots, \sigma_m\}$ of $n$ nodes and $p\in[0,\frac 12]$, then
\begin{equation}\label{NSnuevo}
NS(\sigma_1,\sigma_2, \ldots, \sigma_m) = \frac{1}{(m-1)n(n-1)} K_{cev}^{(p)} (\sigma_1,\sigma_2, \ldots, \sigma_m),
\end{equation}
where $NS(\sigma_1,\sigma_2, \ldots, \sigma_m)$ is the normalized mean strength  of the projected evolutive competitivity network of ${\mathcal R}$ for fixed parameters $\alpha=2$, $p\in[0,\frac 12]$, $q=1$ and $\gamma=0$.
\end{proposition}

\begin{proof}
Note that Remark~\ref{rem:ns-cev} pointed out that we can compute the corrected evolutive Kendall distance with penalty parameter $p$ by the addition of the weights of all the edges of the projected  evolutive competitivity network associated to ${\mathcal R}$, for fixed parameters $\alpha=2$, $p\in[0,\frac 12]$, $q=1$ and $\gamma=0$. From this observation, since the maximal weight of each edge occurs when there are no ties in the rankings of ${\mathcal R}$ (it is in fact $2(m-1)$), then normalized mean strength $NS(\sigma_1,\sigma_2, \ldots, \sigma_r)$ of the projected evolutive competitivity network of ${\mathcal R}$ is
\[
\frac{1}{(m-1)n(n-1)} K_{cev}^{(p)} (\sigma_1,\sigma_2, \ldots, \sigma_m).
\]
\end{proof}

%
%
%
%
%

\section{An application to the competitiveness of IBEX index from 2003 to 2013}\label{sec:economy}

The expected return on an investment and the investment risk are key determinants in the decision of a market economy investor to invest \cite{Brigh2002,Fam93}.

In this section we use the methodology developed in previous sections for analysing the return and the volatility from the stock market prices of the IBEX-35 (Data from Spanish Stock Exchange) corresponding to the twenty-five companies that have been trading on the market contributing to this stock market index during the whole period 2003-2013.

\subsection{Description of the dataset and methods used}

{For} each year we have about $250$ data for each company, corresponding to the daily values. The specific number of trading days for each year is given in Table \ref{ta:ndi}. The data from the IBEX-35 returns and volatilities have been extracted from the {\sf Invertia} website~\cite{invertia}.

\begin{table}[h]
\begin{center}
\begin{tabular}{|c|c|}
\hline
Year  &  Trading days ($N_y$)\\
\hline
$2003$ & $250$ \\
\hline
$2004$  & $251$  \\
\hline
$2006$  &  $254$  \\
\hline
$2007$  &  $253$ \\
\hline
$2008$  & $253$  \\
\hline
$2009$  & $254$ \\
\hline
$2010$  & $256$ \\
\hline
$2011$  & $257$ \\
\hline
$2012$  & $256$ \\
\hline
$2013$  & $255$ \\
\hline
\end{tabular}
\bigskip
\caption{Details of the number of trading days year by year}
\label{ta:ndi}
\end{center}
\end{table}

In both cases (return and volatility) the values have been annualised. Specifically, the (annualised) daily return $R_i$ for the twenty-five considered stocks in the IBEX 35 has been calculated according to the formula
\begin{equation}\label{eq:return}
R_i =N_{y}\left(\frac{SP_i-SP_{i-1}}{SP_{i-1}}\right),
\end{equation}
where $N_{y}$ is the number of trading days of the previous year as it is reflected in Table~\ref{ta:ndi} (the same for every day $i$ of that year), and $SP_{i}$ and $SP_{i-1}$ are the daily Stock Price of that stock in, respectively, the considered day and the day before.

In the studied period the data corresponding to return show a minimum of $-2.8223$ and a maximum of $1.8423$. In the same period, the data corresponding to volatility show a minimum of $0.1013$ and a maximum of $1.0699$. This information is crucial to select  in a proper way the parameter $\Delta x$ which allows us to define a tie between two values when they are close {\em enough}. In order to estimate the risk of an investment, some parameters may be considered. The simplest is the range of the share prices considered along a period of time, i.e., the difference between the greatest and the smallest value during that period of time. Another usual parameters are the variance and the standard deviation. It is worth mentioning that although the most commonly used concept to estimate the risk of an investment is the variance, many academics prefer to use the standard deviation since it offers advantages with respect to the number of decimals we have to employ \cite{Brea2006}. For example, we prefer say that a stock share has an annual performance of $8$ per cent instead of saying that it has a variance of $0.0064$. For this reason, we have employed the standard deviation instead of the variance in our analysis as the risk indicator.

Specifically, the (annualised) daily volatility $V_i$ has been calculated according to the formula
\begin{equation}
V_i=\left(\frac{1}{N(i)}\sum_{j=0}^{N(i)-1} (R_{i-j}-\mu_i)^2 \right)^{1/2},
\end{equation}
where $N(i)$ is the number or Trading Days from day $i$ (of the year $k$) to the same day of the previous year ($k-1$), $R_i,R_{i-1},\cdots,R_{i-N(i)+1}$ are the daily return of all the days Trading Days from day $i$ (of the year $k$) to the same day of the previous year ($k-1$) and $\mu_i$ is the mean of these values (also called {\sl annualized moving average of the daily return} at day $i$).

\subsection{Analysis and results}
It is known that Financial Market is a chaotic system in the sense that it is sensitive to initial conditions and gives rise to effectively unpredictable long-term behaviour \cite{savit,Peters91}. A possible reason for this development may be that it contains positive feedback (which tends to amplify trends over time) and negative feedback (which tends to reduce trends over time). However, some attempts to detect the presence of chaos in certain systems have not had any success. For example, in~\cite{Abh95} the authors tested for the presence of chaos in the FTSE 100 Index using a six month sample of about 60,000 minute by-minute returns and found little to support the view that the process is chaotic at any frequency. A question to investigate is the evidence of chaos in the evolution of some other parameters related to the stock markets which give information about the stock of the different companies.

We have investigated the evolution of competitivity for the daily values of return and volatility for the twenty-five considered stocks in the IBEX 35 during the period 2003-2013. The analysis is based on the study of global (structural) properties of the corresponding competitivity networks. As the  theoretical framework established in previous sections  differs from those traditionally used in the analysis of stock markets, the information we get is completely different from that obtained through a traditional analysis, since we have analysed the ranking fluctuations of the stock values rather than the actual values.  In particular, we have considered the evolution of the Normalized Mean Strength ($NS$) for the competitivity networks obtained from return and volatility during the period 2003-2013, but many other structural, global or local, parameters could be considered as it happened in \cite{Cri13}. $NS$ is a good indicator of the number of fluctuations in the daily rankings, since by using Proposition~\ref{prop_ns-cev}, the higher $NS$ is the more fluctuations in the rankings of return or volatility we get. Furthermore, since a high number of  fluctuations in the daily rankings of return or volatility can be economically understood as uncertainty in the investment strategies or non-preconditioned markets, $NS$ could be considered as a quantitative indicator of such economical variables.

Once we have fixed the economical dataset to be analysed (in this case the ranking fluctuations of the return and volatility of the common twenty-five companies that have been trading on the Spanish Stock Market contributing to this stock market IBEX35 index during the whole period 2003-2013), the rankings to be considered must be fixed. Note that the rankings are given by the return or the volatility and therefore they are actually rankings with ties. Furthermore, these ties must be also considered as {\sl approximated ties} as we introduced in Definition~\ref{def:aprox}. Note that either return and volatility are non-integer scores and therefore the use of approximated ties helps avoiding the instability of ties due to round-off errors. In addition to this, some stock markets values which are very similar (but different) could be considered equivalent for an investor, and therefore they are {\sl approximately tied}.

In order to fix the setting of the analysis the precision threshold $\Delta x$ must be stated. The right choice of $\Delta x$ should consider the coexistence of tied and non-tied nodes, since a very low value of $\Delta x$ produces unstable results due to round-off errors, while a very high value of $\Delta x$ forces to all the nodes to be tied and therefore produce a rough analysis. The trade-off between number of ties, accuracy and stability of the analysis is obtained for values of $\Delta x$ that produce a significant but not too high number of ties. Figure~\ref{figtiesall} shows the relationship between $\Delta x$ and the number of ties for the rankings of the return (panel~(a)) and volatility (panel~(b)) for the IBEX index along 2003-2013. This figure illustrates that there is a phase transition for values $\Delta x$ between $0.02$ and $1$ and therefore a good choice of $\Delta x$ could be $\Delta x=0.05$.

\begin{figure}[h!]
\begin{center}
 \includegraphics[width=0.49\textwidth]{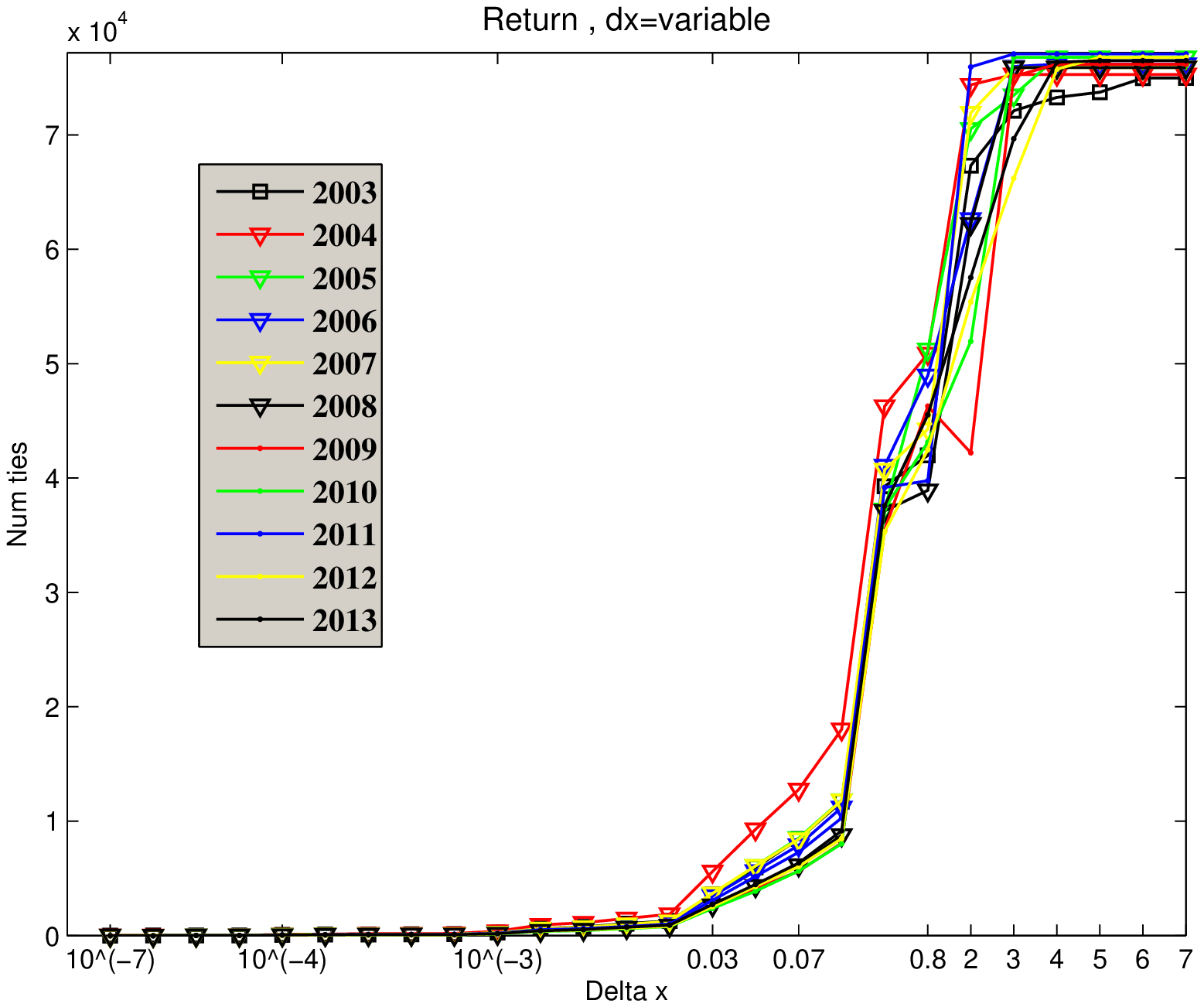}
 \includegraphics[width=0.49\textwidth]{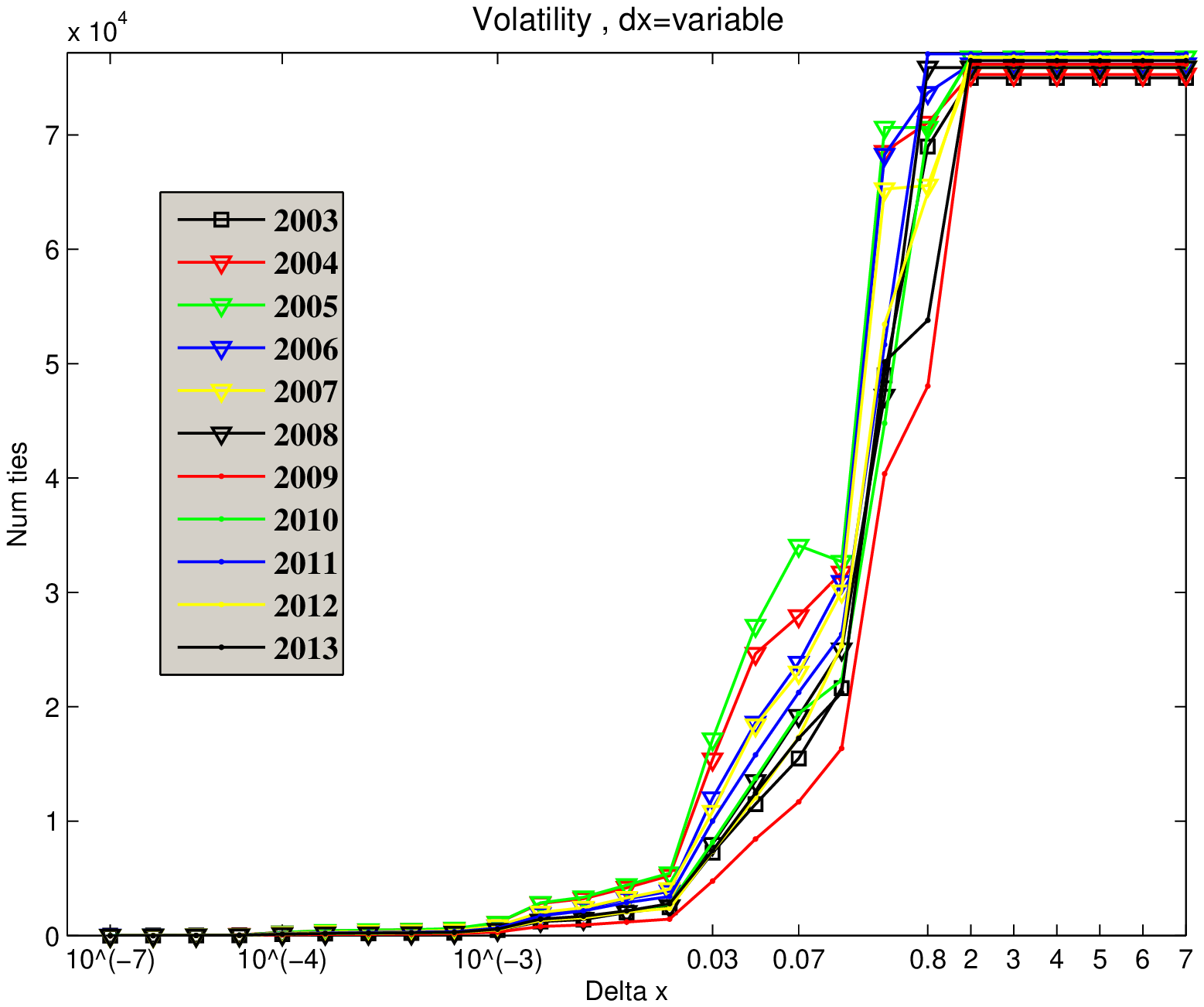}
\end{center}
\caption{Number of ties versus $\Delta x$ resulting for the rankings given by the return (panel (a) or the volatility (panel (b)) along 2003-2013}
\label{figtiesall}
\end{figure}

$\Delta x=0.05$ produces a not too high number of ties and also makes that $K_{cev}^{(p)}$ exhibits different values from $K_{ev}^{(p)}$, as Figure~\ref{figFN1} shows. Figure~\ref{figFN1} presents the values of $K_{cev}^{(p)}$ and $K_{ev}^{(p)}$ for the rankings of the return (panel~(a)) and volatility (panel~(b)) for the IBEX index  along 2003 in terms of $\Delta x$. Note that either for very high or for very low values of $\Delta x$, we get that $K_{cev}^{(p)}\approx K_{ev}^{(p)}$, as it is shown in Figure~\ref{figFN1}. This is due to the following facts: {\it (i)} low values of $\Delta x$ makes a very low number of ties and therefore a very low number of long-term crossings that fix to the Case 4 in Definition~\ref{def:Kevdc}, {\it (ii)} high values of $\Delta x$ makes that all the nodes are always tied with high probability and therefore the number of  long-term crossings that fix to the Case 4 in Definition~\ref{def:Kevdc} is very low. The value $\Delta x=0.05$ produces a non-negligible number of ties (see Figure~\ref{figtiesall}) and makes that   $K_{cev}^{(p)}$ exhibits different behaviour to $K_{ev}^{(p)}$ (see Figure~\ref{figFN1}), so we will fix this value for the analysis.

\begin{figure}[h!]
\begin{center}
  \includegraphics[width=0.49\textwidth]{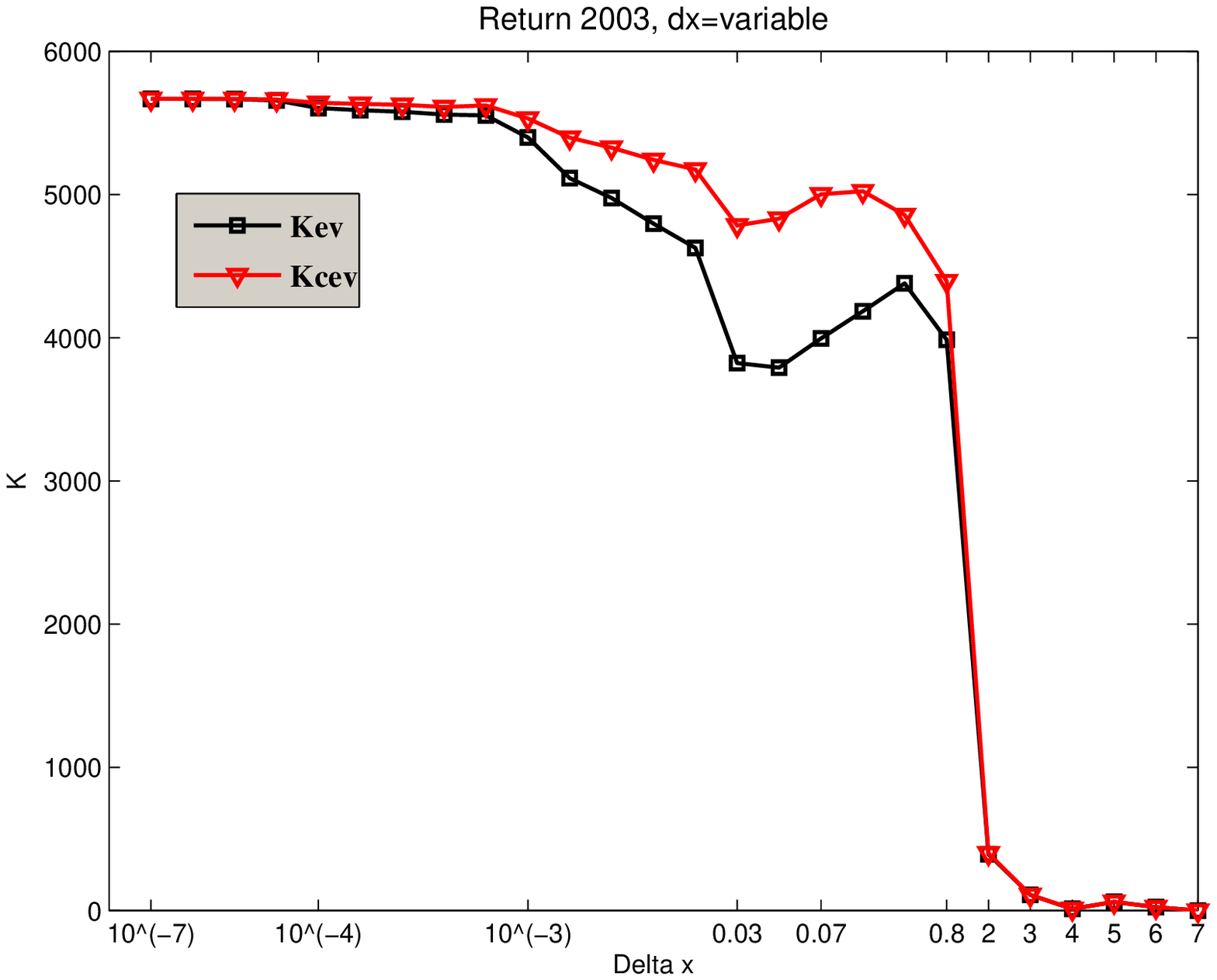}
  \includegraphics[width=0.49\textwidth]{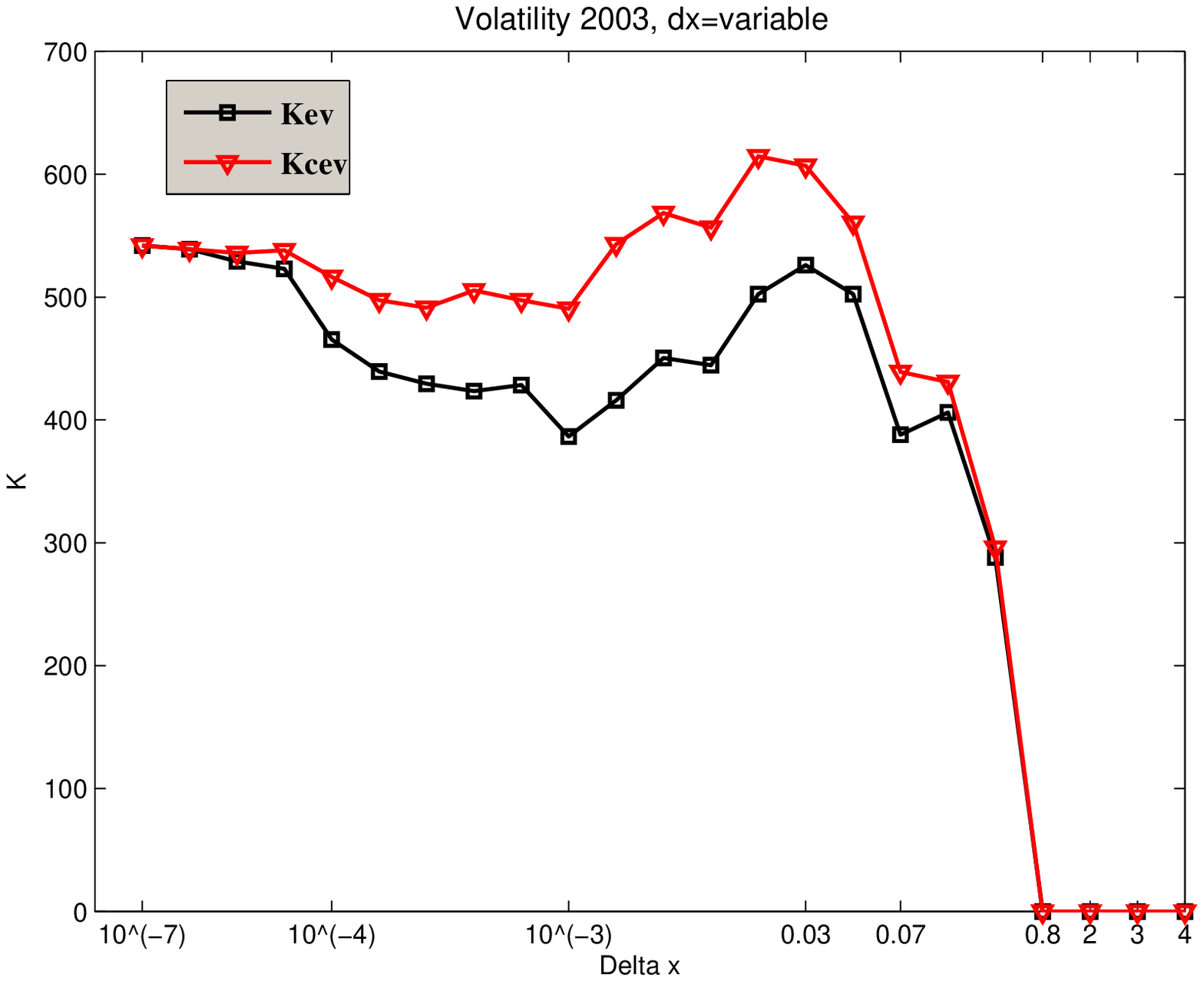}
\end{center}
\caption{  $K^{(p)}_{cev}$ (in red)  and $K^{(p)}_{ev}$ (in black) versus $\Delta x$ for the rankings given by the return (panel (a) or the volatility (panel (b)) resulting from year 2003, with $p=0.5$}
\label{figFN1}
\end{figure}

It is quite remarkable that the choice of $\Delta x$ could deeply modify the results obtained in the analysis. Figure~\ref{figEvolNSvsYears} shows the evolution of the Normalized Mean Strength ($NS$) for the competitivity graphs of return (top panels) and volatility (bottom panels) and three different values of $\Delta x=0.5,\,0.05,\,0.005$. The differences are stronger when we change $\Delta x$ from $0.5$ to $0.05$  than when we move from $\Delta x=0.05$ to $\Delta x=0.005$ . These differences also affect the tendencies, which is specially visible in the case of the evolution of the competitivity graphs of volatility.

\begin{figure}[h!]
\begin{center}
  \includegraphics[width=0.3\textwidth]{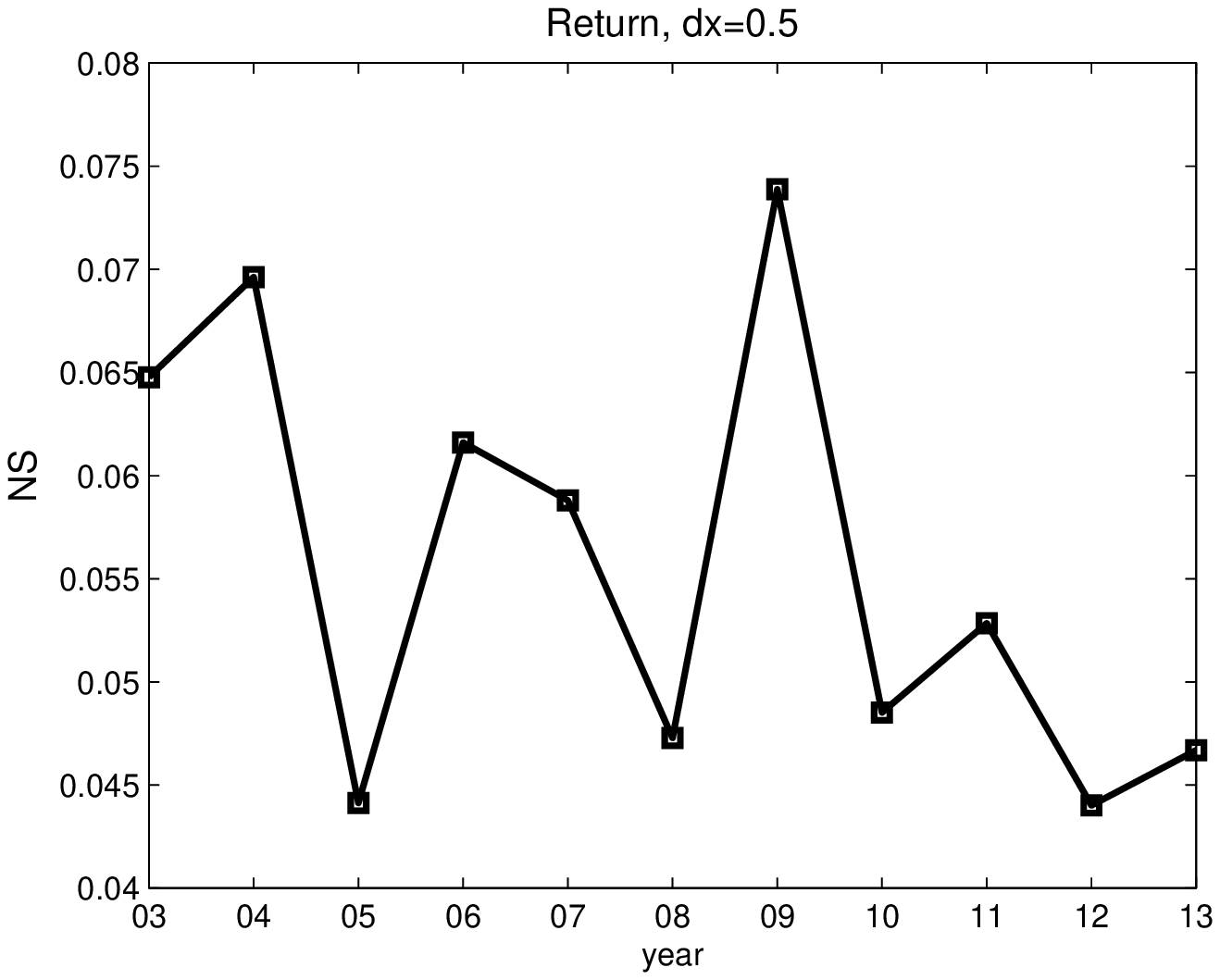}\quad
  \includegraphics[width=0.3\textwidth]{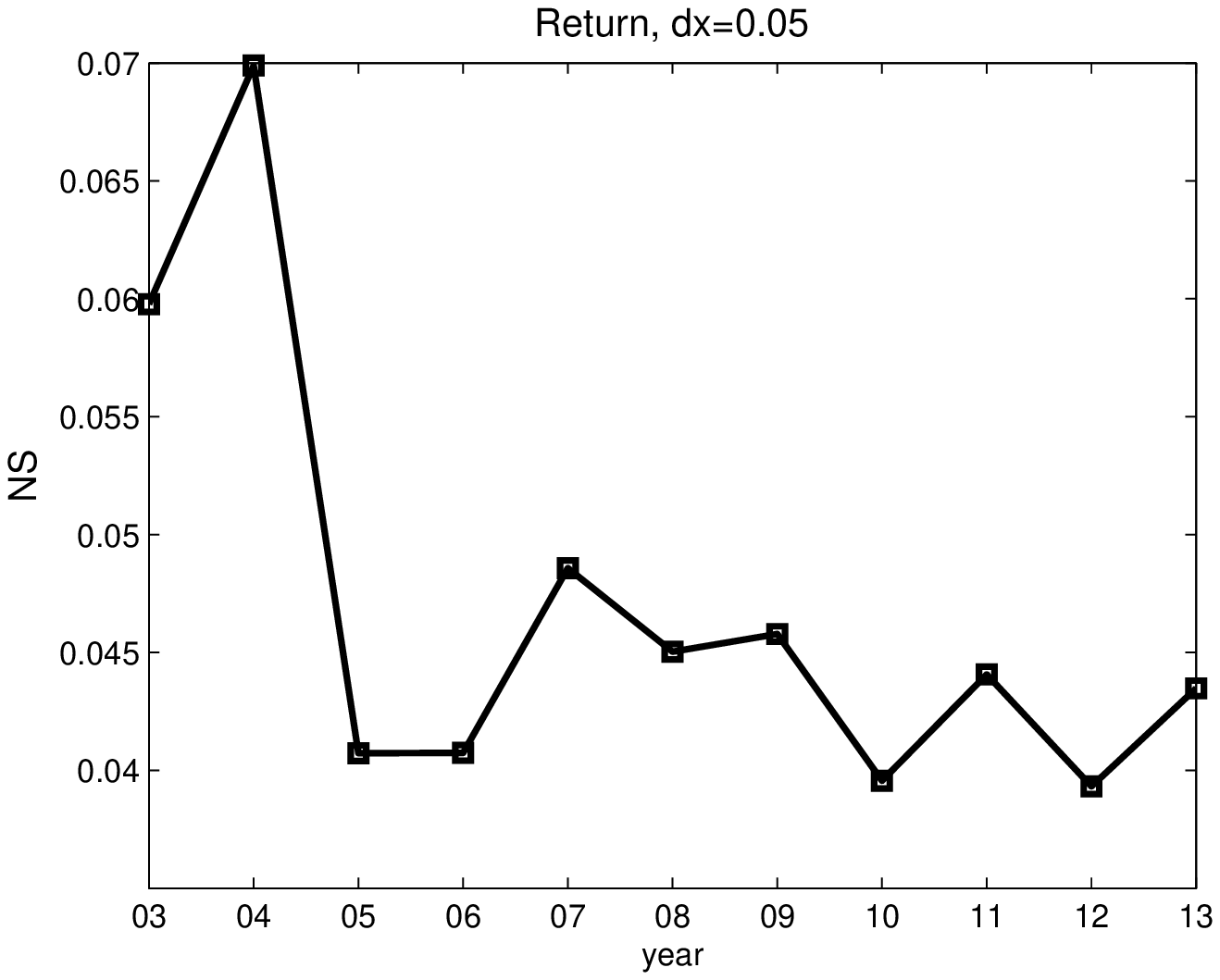}\quad
  \includegraphics[width=0.3\textwidth]{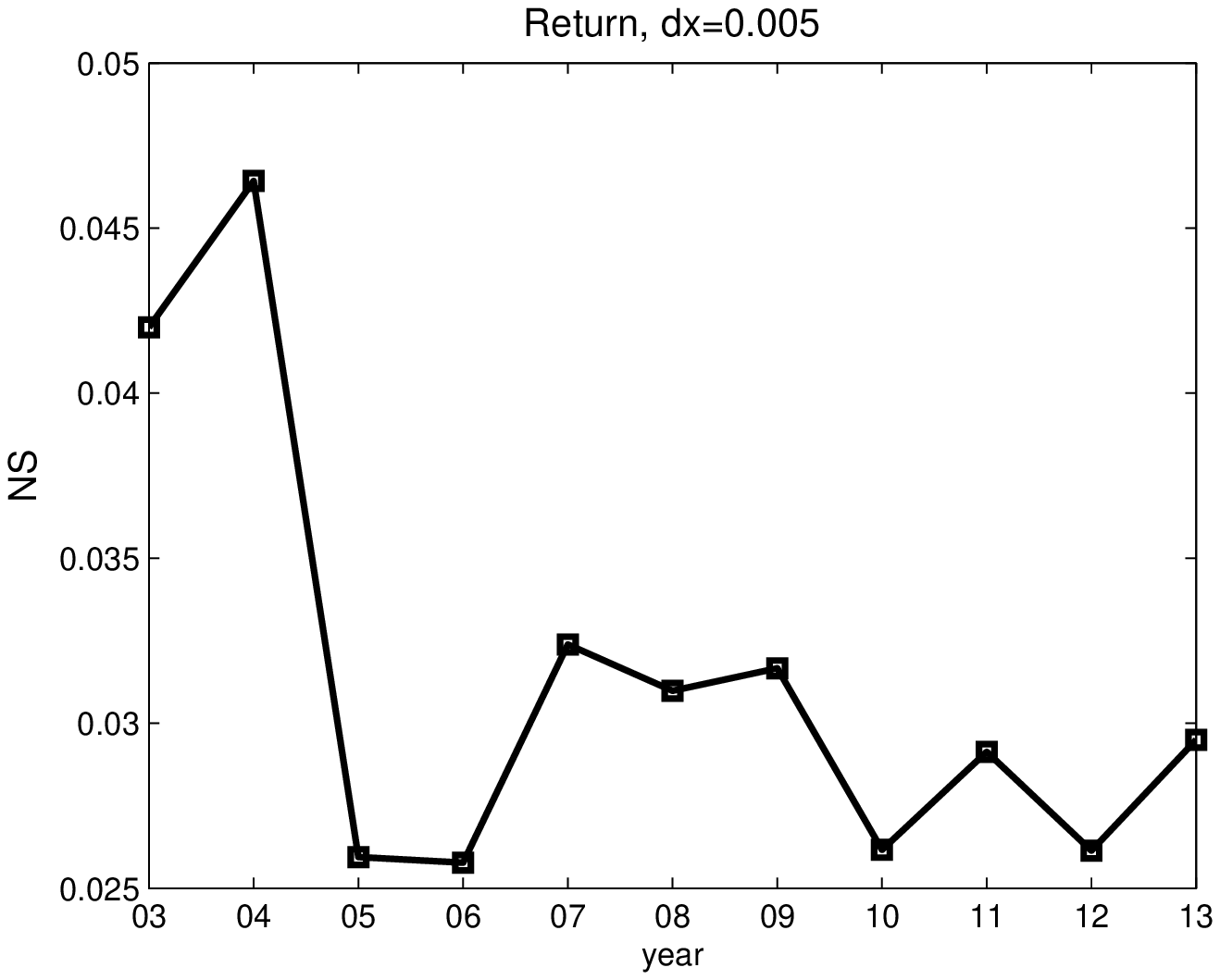}\\
  \includegraphics[width=0.3\textwidth]{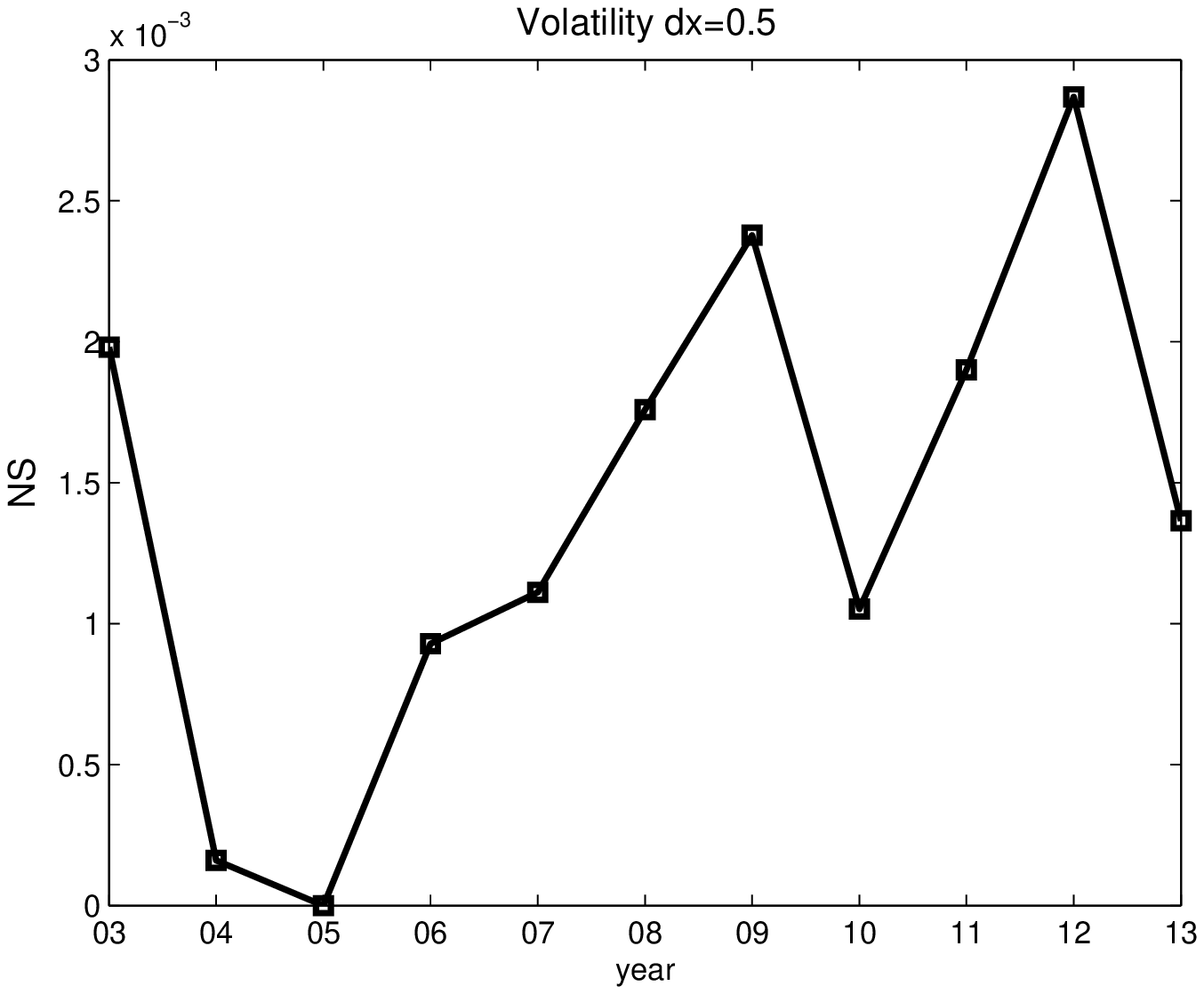}\quad
  \includegraphics[width=0.3\textwidth]{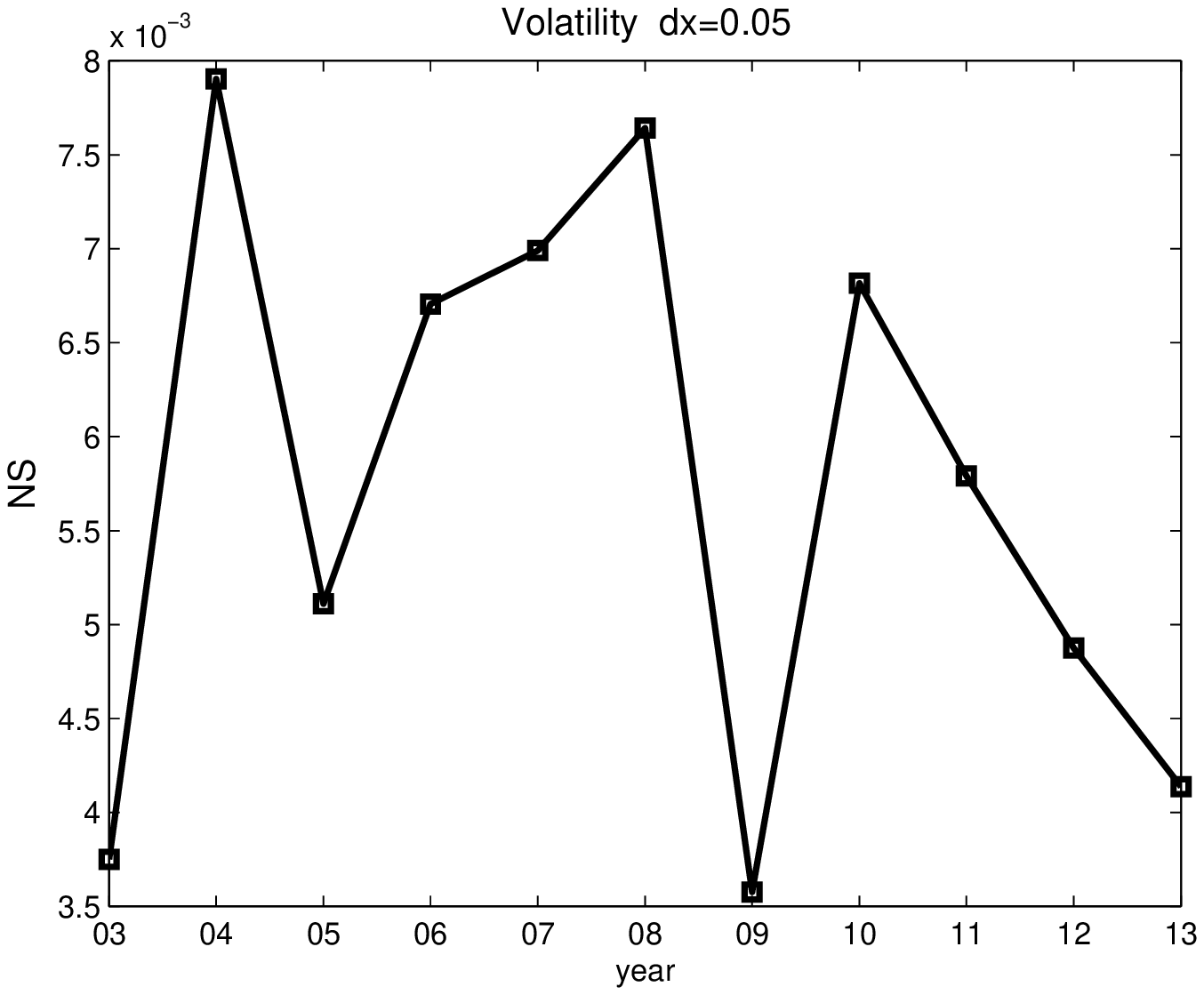}\quad
  \includegraphics[width=0.3\textwidth]{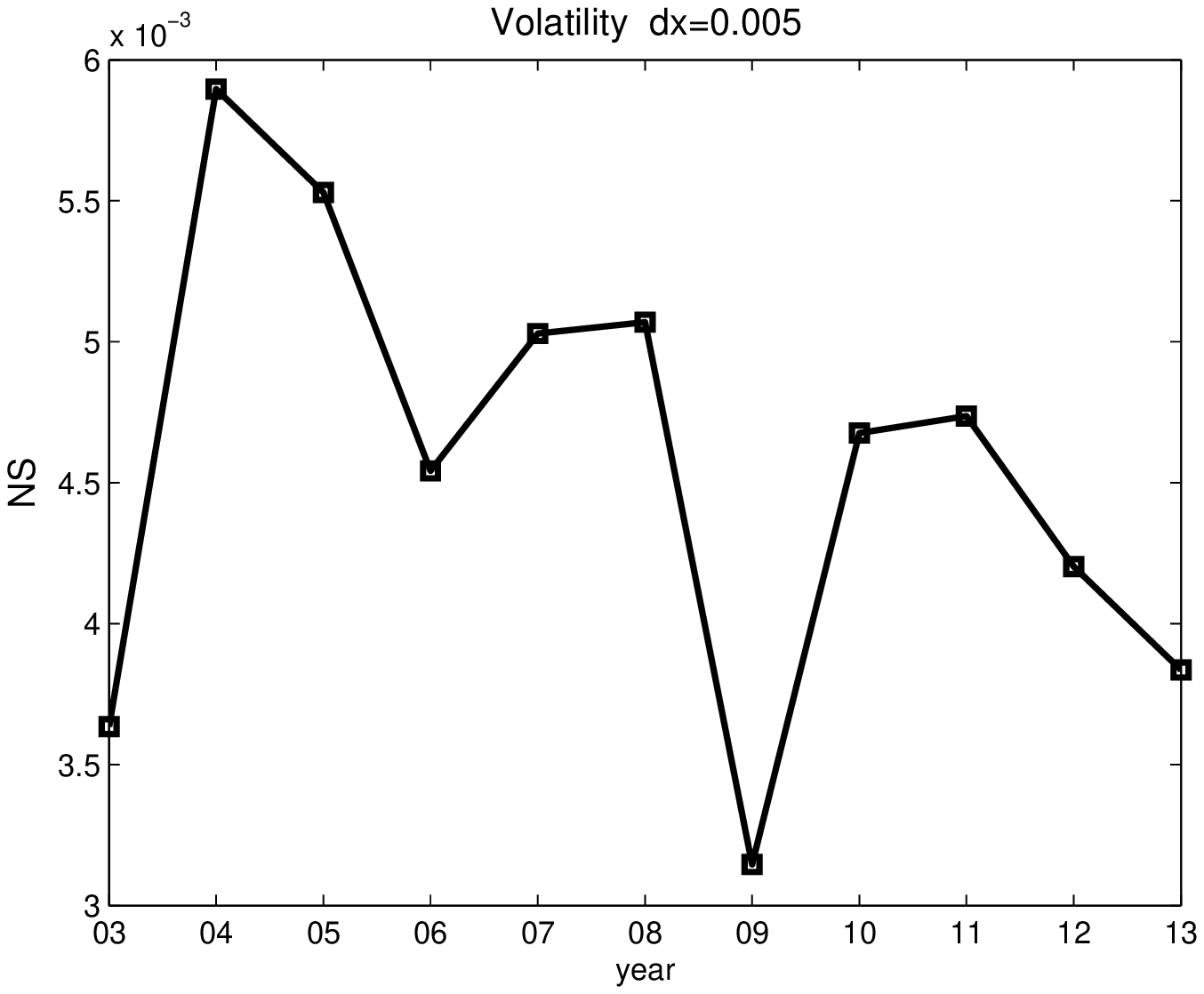}
\end{center}
\caption{Evolution of the Normalized Mean Strength ($NS$) for the competitivity graphs of return (top panels) and volatility (bottom panels) during the  period 2003-2013 and for $\Delta x=0.5$ (left panels), $\Delta x=0.05$ (central panels) and $\Delta x=0.005$ (right panels)}
\label{figEvolNSvsYears}
\end{figure}

A similar situation occurs when the evolution of the Normalized Mean Strength ($NS$) for the competitivity graphs along each year is considered. Figure~\ref{figEvolNSeachYears} shows the evolution $NS$ for the competitivity graphs of return (top panels) and volatility (bottom panels) and three different values of $\Delta x=0.5,\,0.05,\,0.005$ in each year of the period 2003-2013. In these cases, different values of $\Delta x$ produce quite different tendencies in each year and also a different rankings according to $NS$.

\begin{figure}[h!]
\begin{center}
  \includegraphics[width=0.3\textwidth]{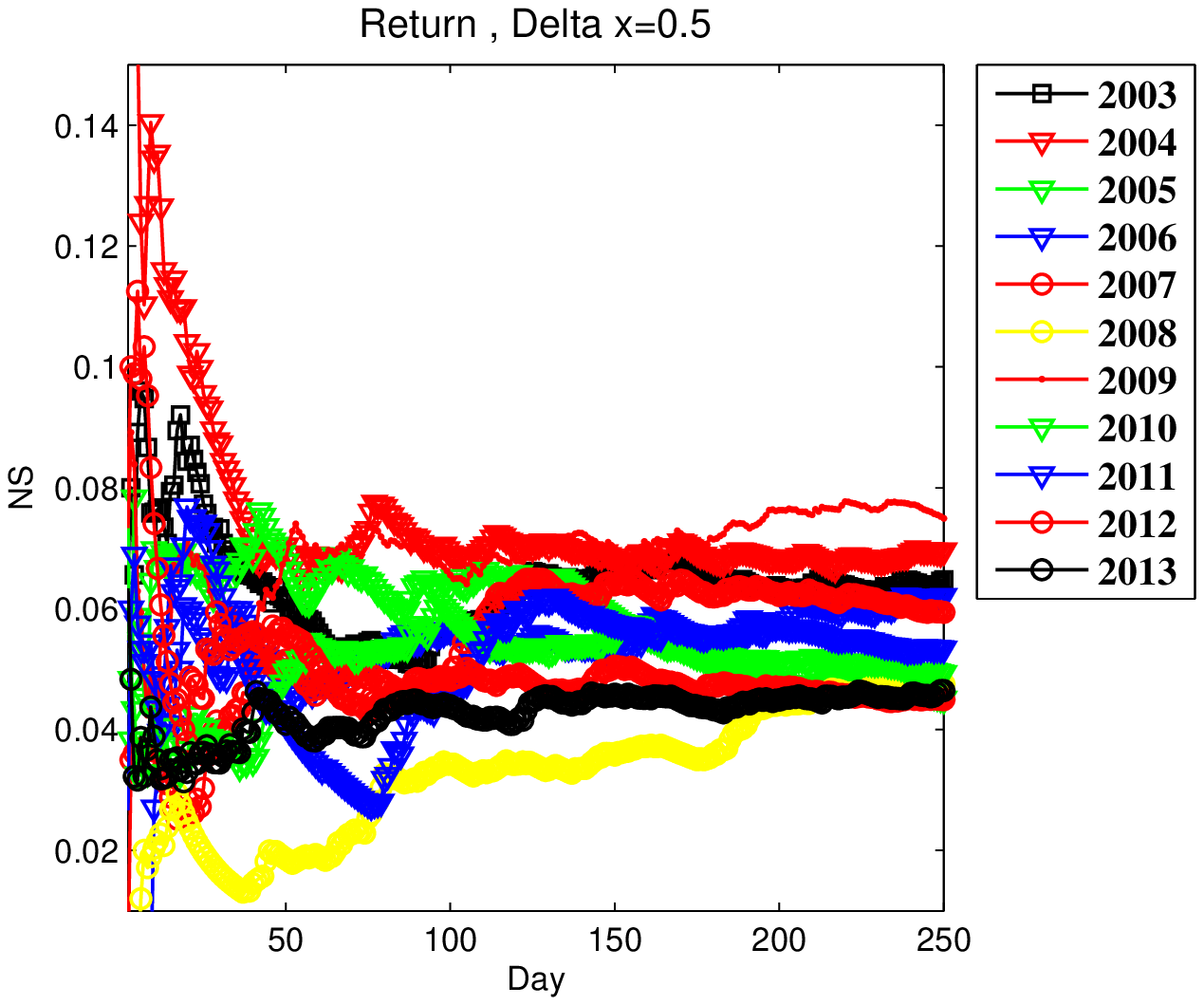}\quad
  \includegraphics[width=0.3\textwidth]{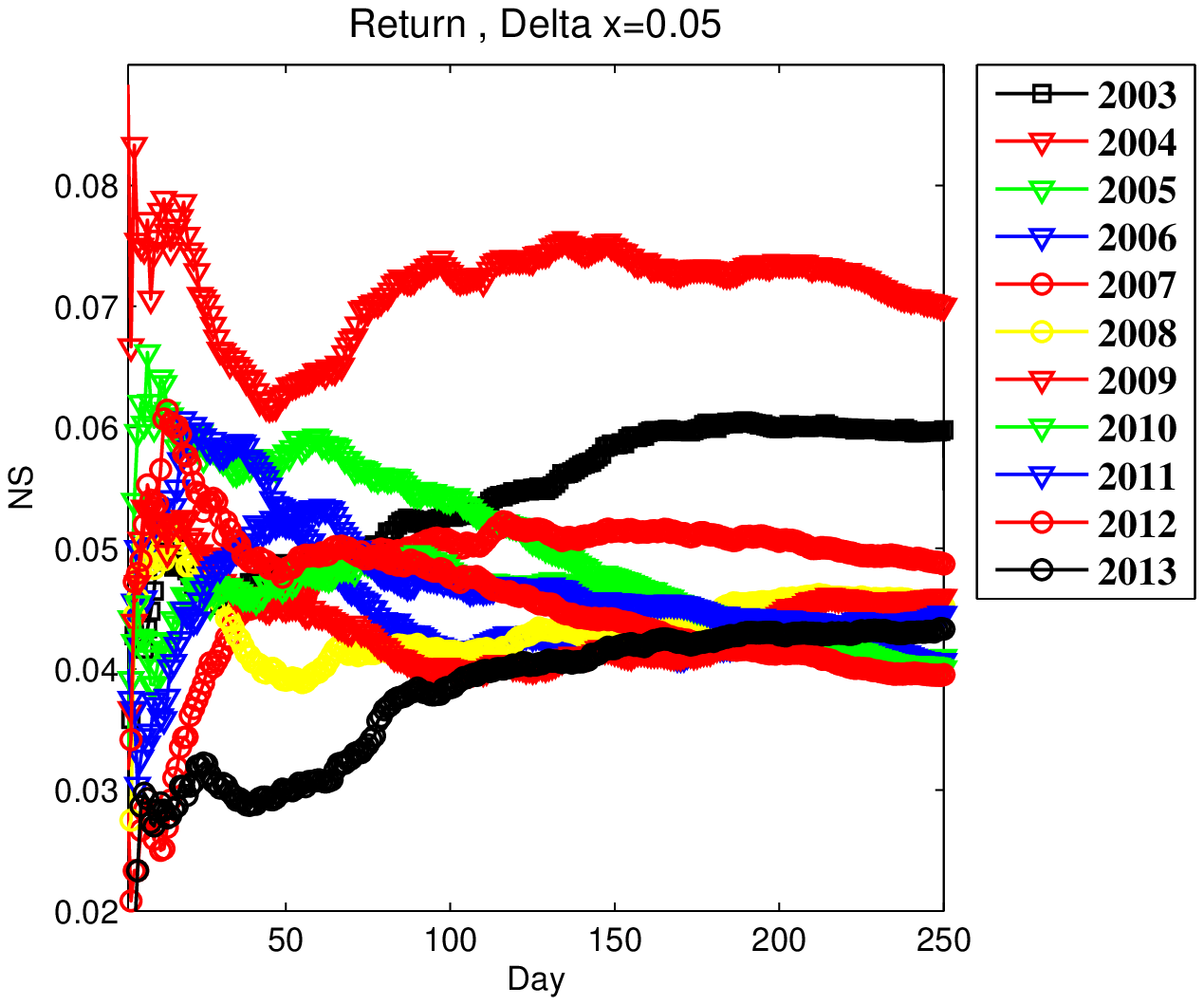}\quad
  \includegraphics[width=0.3\textwidth]{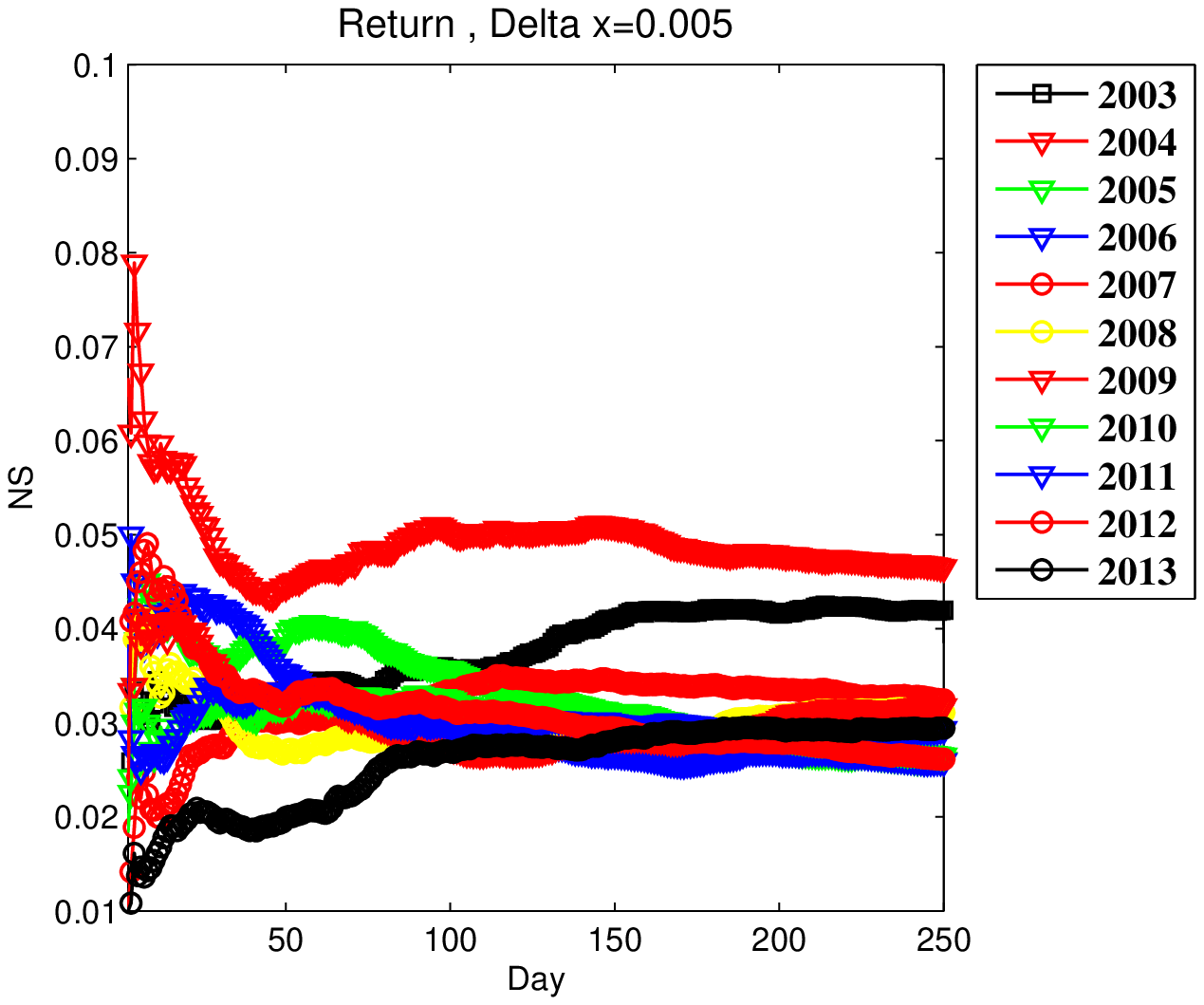}\\
  \includegraphics[width=0.3\textwidth]{NS_2003_2013_5.eps}\quad
  \includegraphics[width=0.3\textwidth]{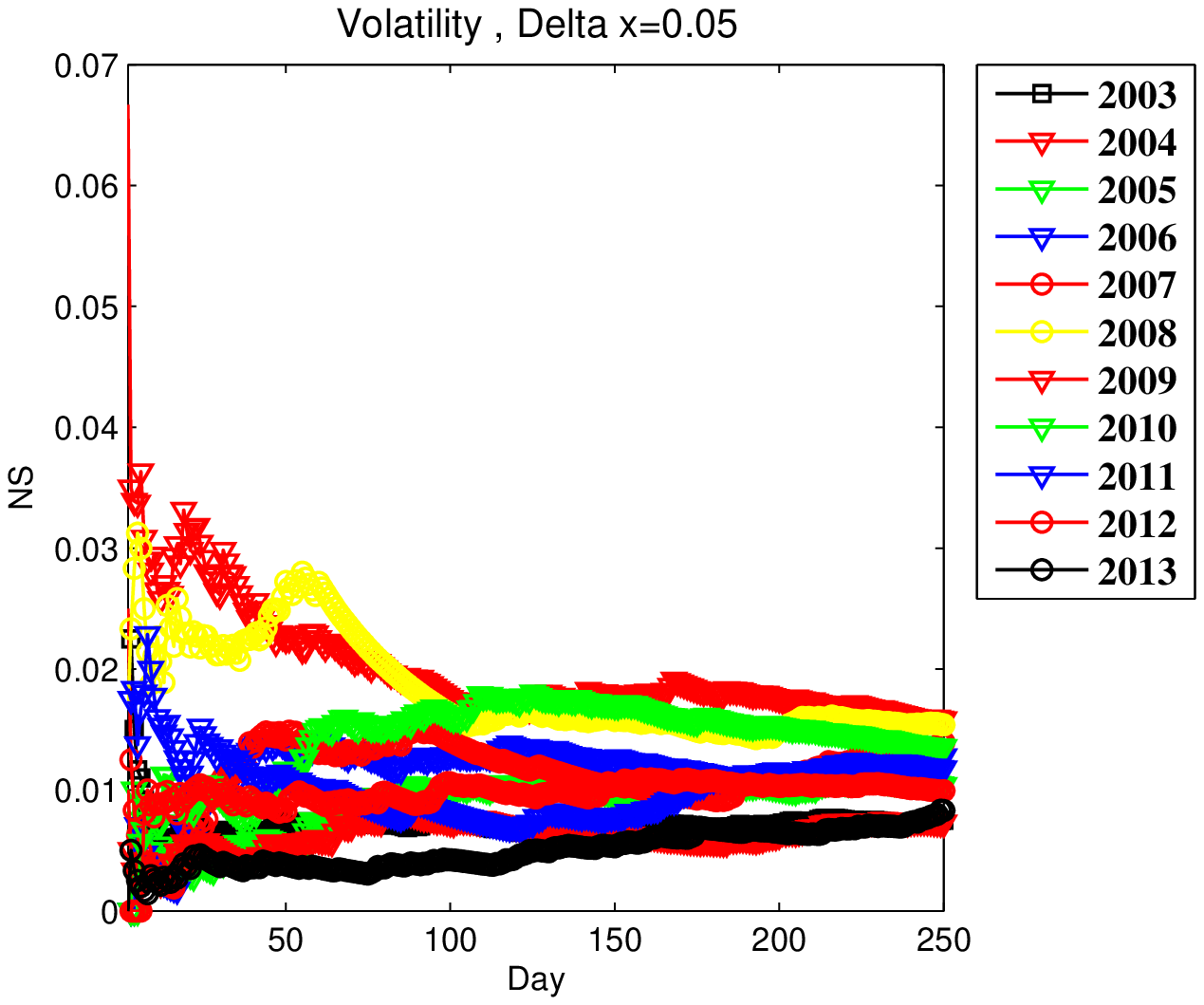}\quad
  \includegraphics[width=0.3\textwidth]{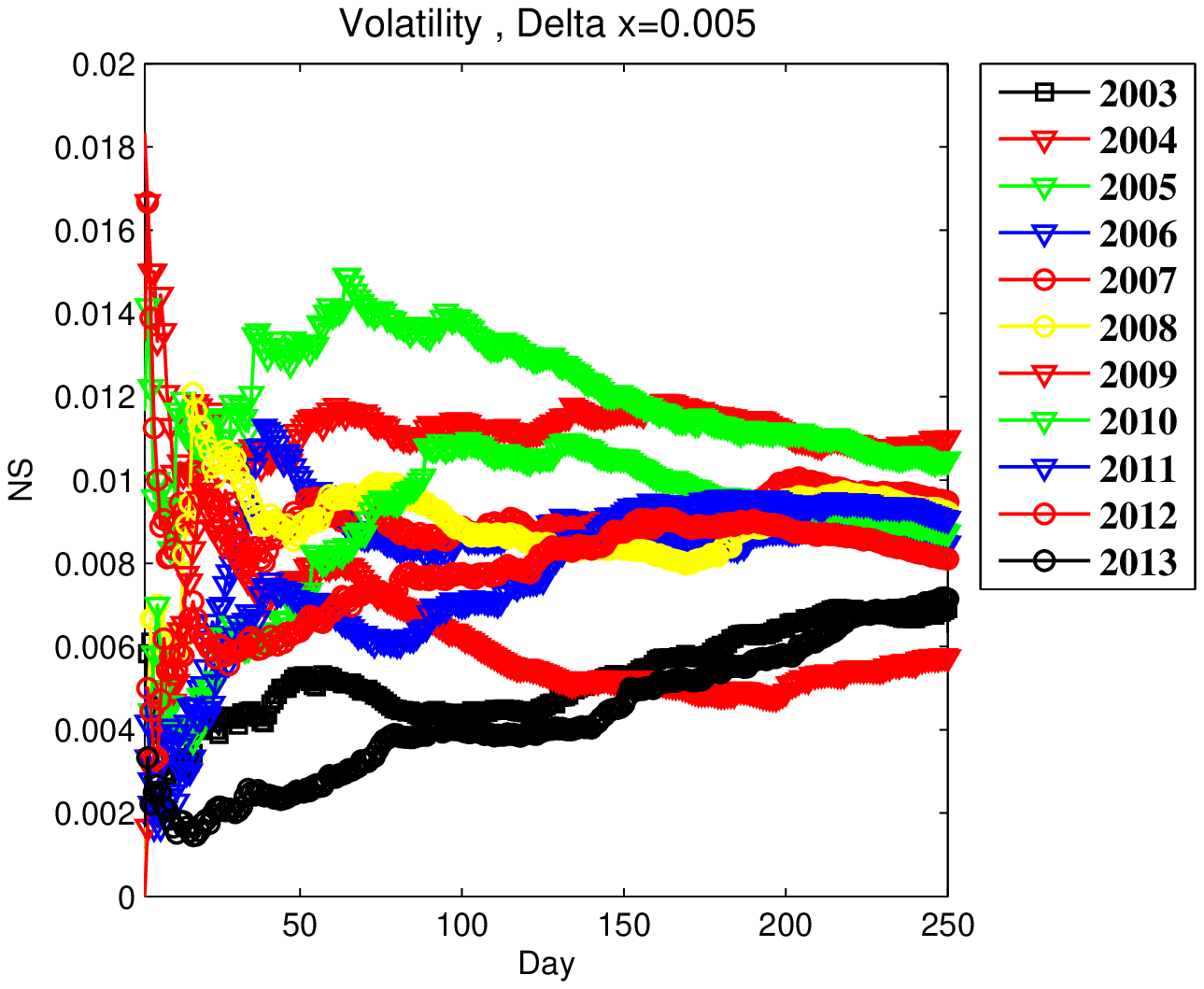}
\end{center}
\caption{  Evolution of the Normalized Mean Strength ($NS$) for the return (top panels) and volatility (bottom panels) during the period 2003-2013 and for $\Delta x=0.5$ (left panels), $\Delta x=0.05$ (central panels) and $\Delta x=0.005$ (right panels)}
\label{figEvolNSeachYears}
\end{figure}

The information encapsulated in the competetivity graphs of the return and volatility is different to the information included in the classic analysis of return and volatility themselves. While the classic and widely accepted studies about the evolution of the return and volatility measures the global changes in the markets, they do not take into account the intrinsic fluctuations of the companies that have been trading on the stock market. This fact is plotted in Figures~\ref{fig2004},~\ref{fig2008}~and~\ref{fig2013}, where the evolution of return, volatility and the  Normalized Mean Strength ($NS$) for the competitivity of the return and volatility  during 2004, 2008 and 2013 are presented. We have chosen these three years since they correspond to three quite different instances of the economical cycle. While 2004 was and expansive year for the Spanish Stock Market, 2008 was the year of the bankruptcy of Lehman Brothers Holdings Inc. that witnessed  the first whole year of one of the deepest economical crisis in Spanish Stock Market and 2013 was the last year with complete data (and it should be the staring point of the recovering process for the Spanish Economy).

According to the Official Reports of the Madrid Stock Market (see \cite{bolsamadrid}), 2004 was an excellent year for the Spanish markets. Despite the economical 	uncertainty caused by the volatility of the petroleum prices, the brilliant results of the major Spanish Companies boosted the Spanish Stock Markets that ended 2004 with profits beyond $17\%$, going ahead of   Wall Street (New York), the German Stock Market, London and Euronext.  As a consequence, this year ended with high return and low volatility, as it is shown in Figure~\ref{fig2004}, panels (a) and (b) respectively. If we have a look at the evolution of $NS$ of the competitivity of return and volatility, we can see a quite stable situation (see Figure~\ref{fig2004}, panels (c) and (d) respectively) around the values $NS\approx 0.037$ for the competitivity of the return and $NS\approx 0.01$ for the competitivity of the volatility.

\begin{figure}[h!]
\begin{center}
  \includegraphics[width=0.48\textwidth]{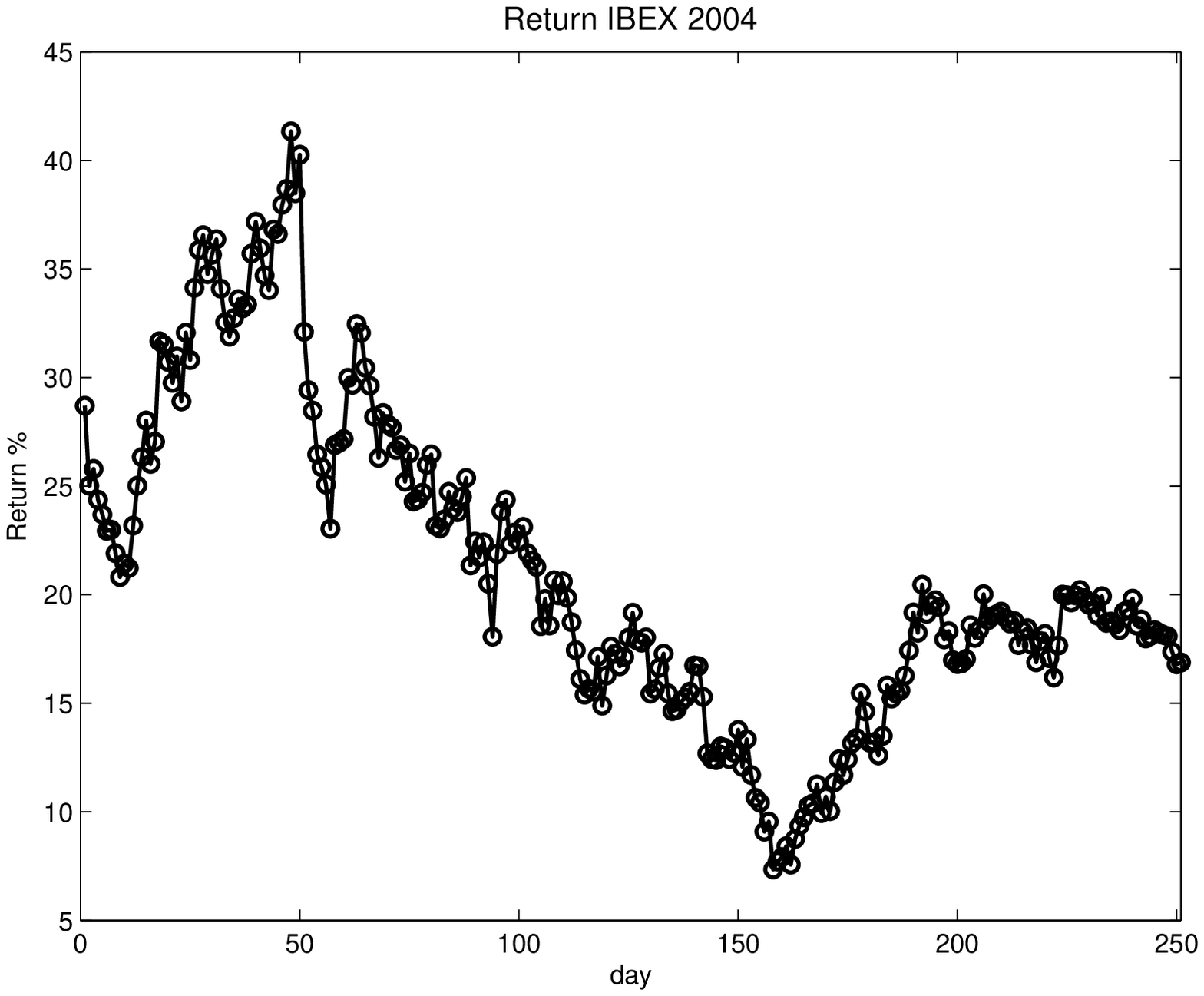}
  \includegraphics[width=0.48\textwidth]{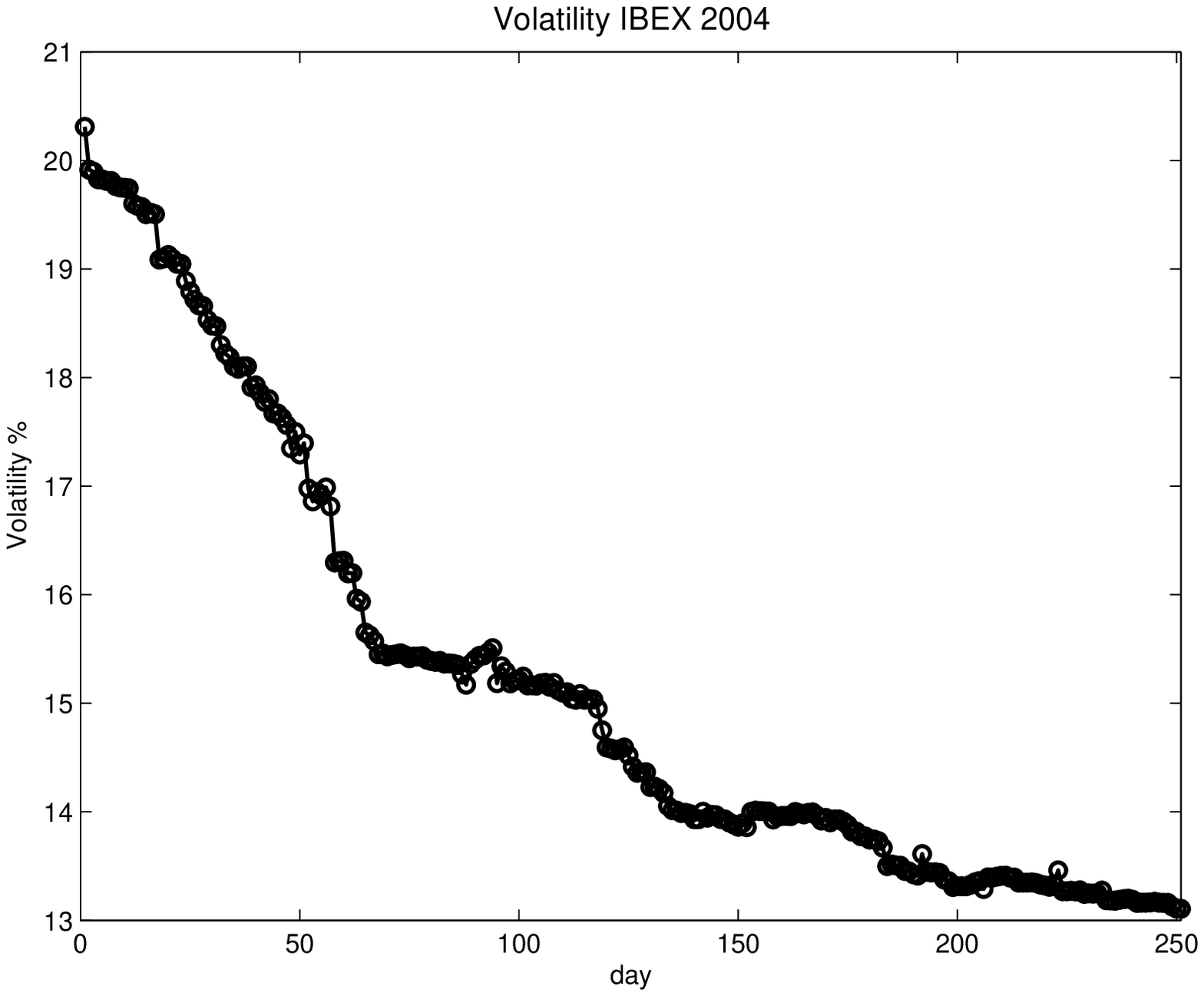}\\
  \includegraphics[width=0.48\textwidth]{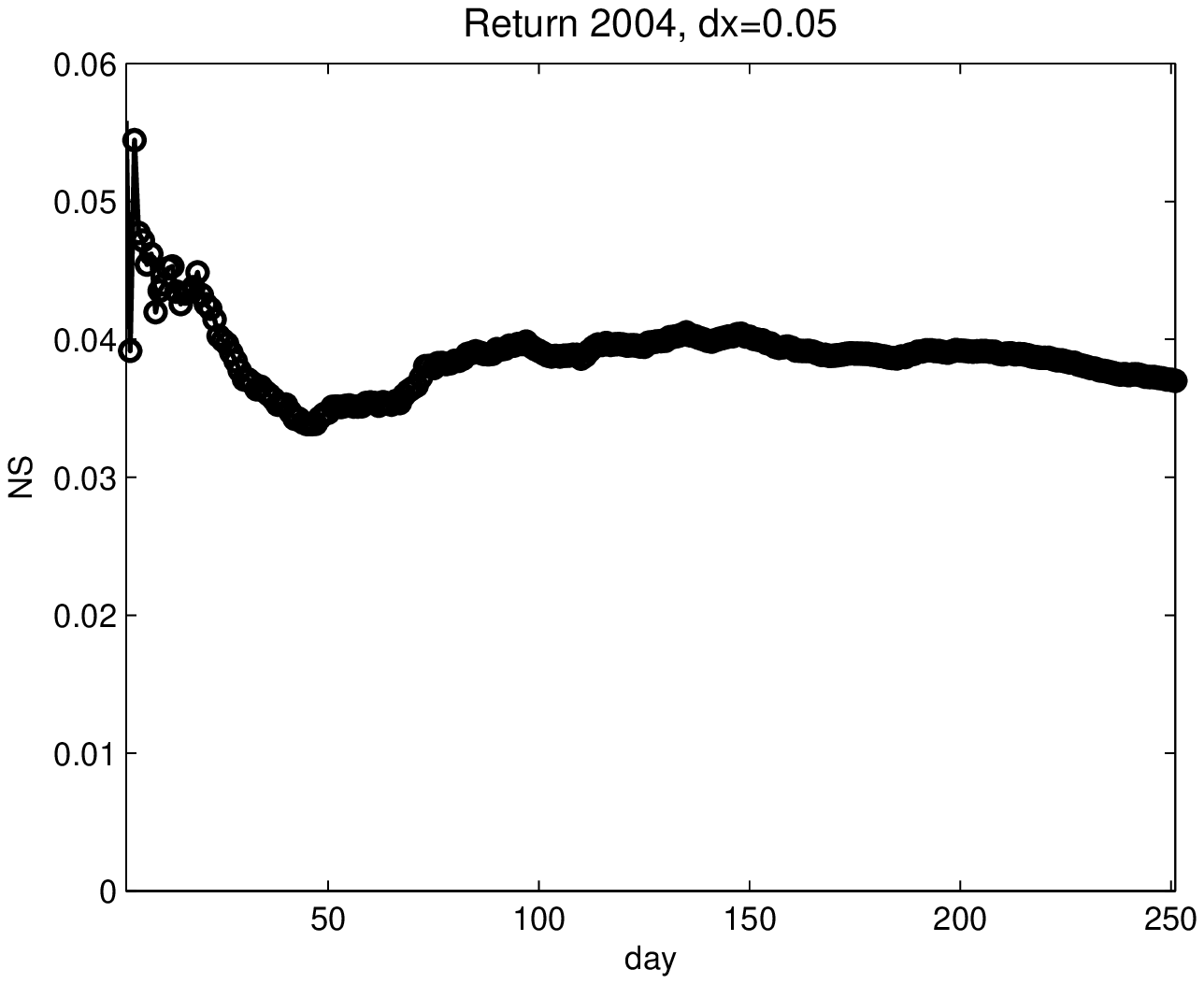}
  \includegraphics[width=0.48\textwidth]{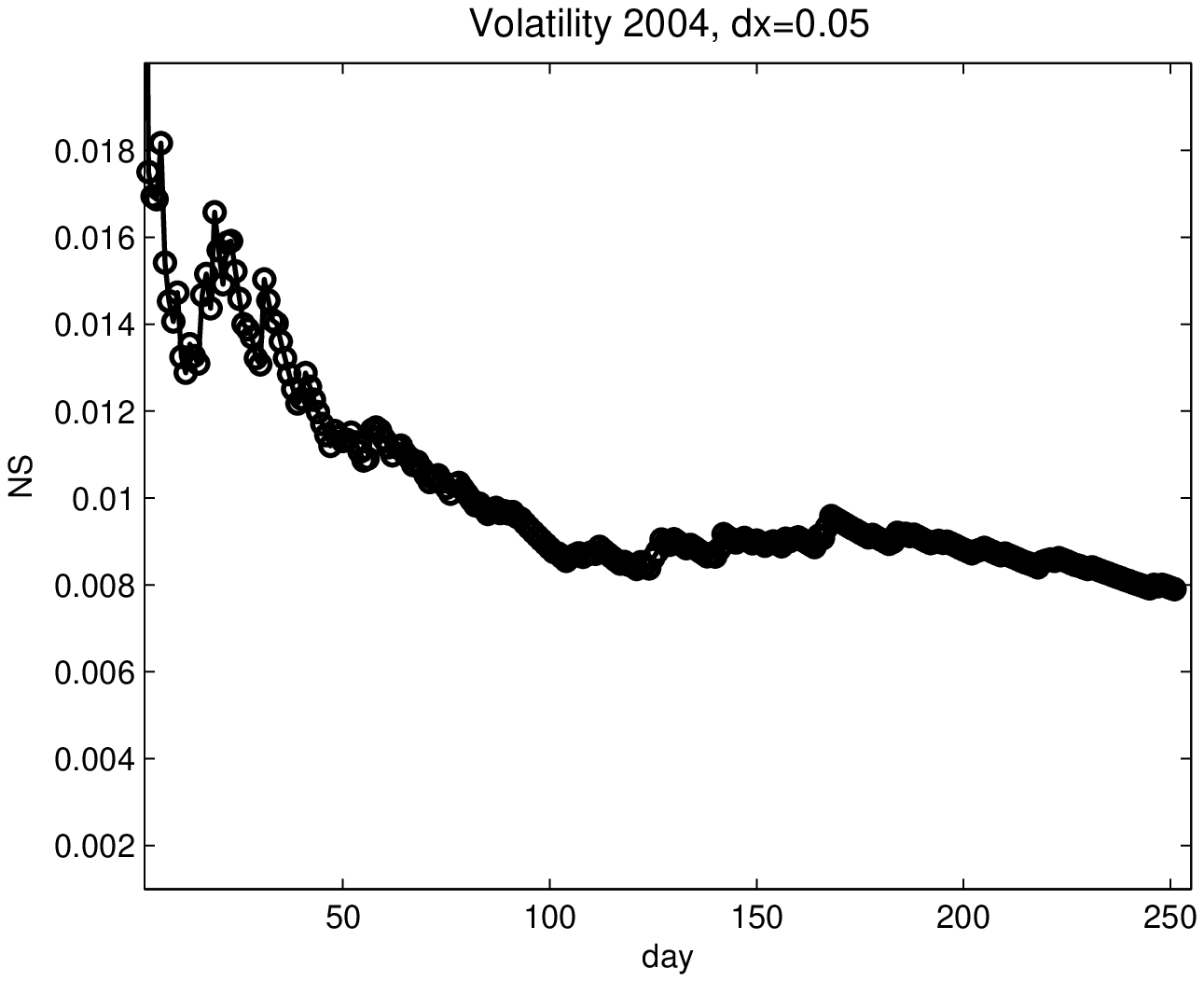}
\end{center}
\caption{Evolution of the Return of IBEX (panel (a)) Volatility of IBEX (panel (b)) Normalized Mean Strength ($NS$) for the competitivity of the return (panel (c)) and for the competitivity of the volatility (panel (d)) during 2004 and for $\Delta x=0.05$}
\label{fig2004}
\end{figure}

On a completely different scenario, 2008 was the worst year for the Spanish Stock Markets. The period from the summer 2007 to the latest weeks of 2008 was one of the most negative and complex moments of the financial recent history worldwide. The economical crisis originated in the American Property and Financial Markets infected the worldwide financial markets setting off mistrust on the economical agents. Official Reports of the Madrid Stock Market (see \cite{bolsamadrid}) point out that the IBEX35 drops around $40\%$ this year and the volatility suffered an amazing increase getting their maximal values in 20 years. These facts are shown in Figure~\ref{fig2008}, panels (a) and (b). The evolution of $NS$ for the competitivity graphs of return and volatility suffered a significant change in their tendencies along the first two months of 2008. On the one hand, $NS$ for the competitivity graphs of return fell at the end of February but tried recovering along the rest of the year, getting values around $NS\approx 0.023$ (see Figure~\ref{fig2008}, (c)), which means that fluctuations of the ranking obtained for the return were reduced during the first two months of the year but this ranking changed more and more since then. On the other hand, $NS$ for the competitivity graphs of volatility started the year quite erratic, but it dropped between day 60 and 100, and it couldn't recover its value during the rest of the year. This fact can be understood as the fact that the ranking obtained for the volatility was very rigid and stable  from April until the end of the year with values  $NS\approx 0.008$ (see Figure~\ref{fig2008}, (d)).

\begin{figure}[h!]
\begin{center}
  \includegraphics[width=0.48\textwidth]{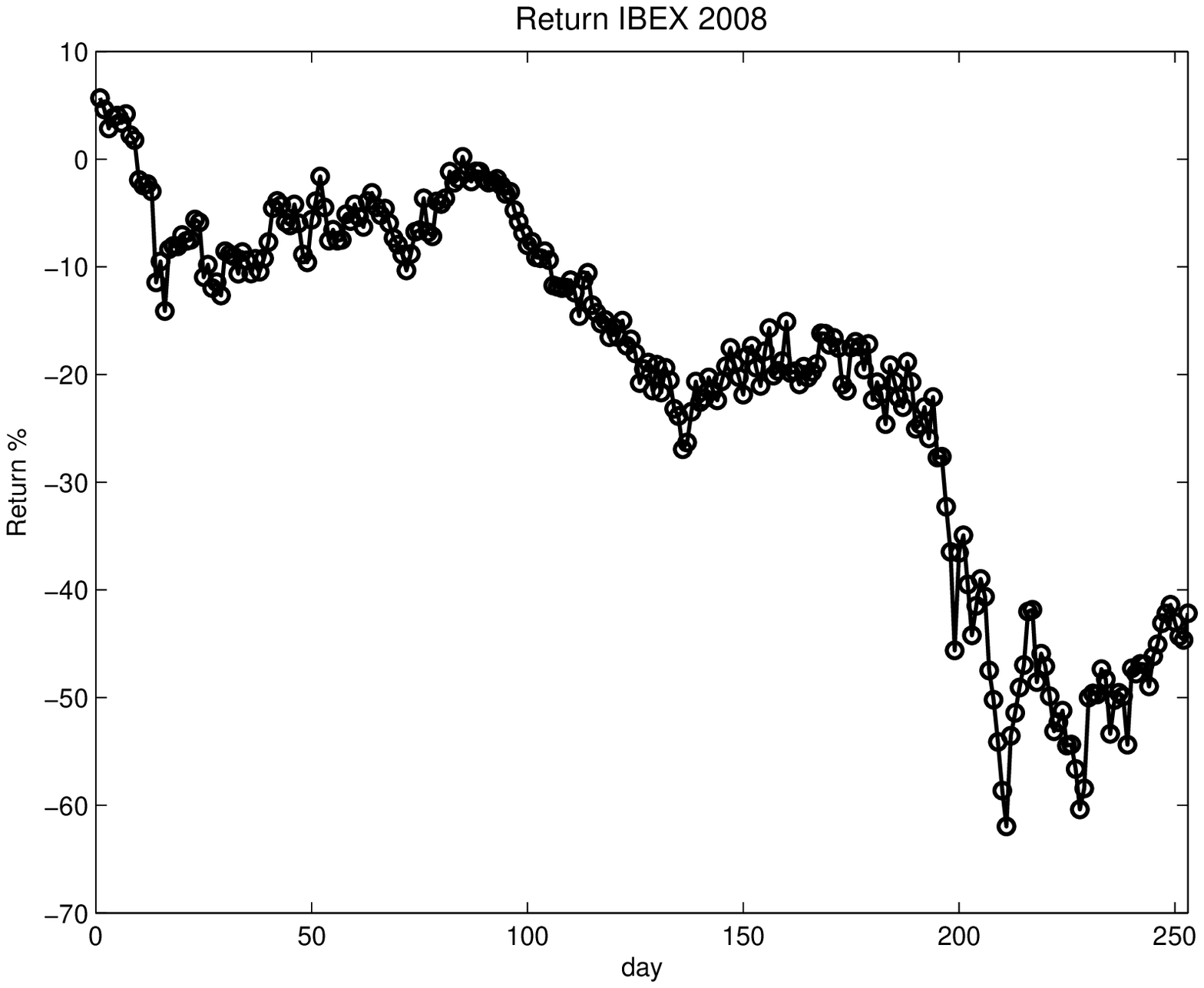}
  \includegraphics[width=0.48\textwidth]{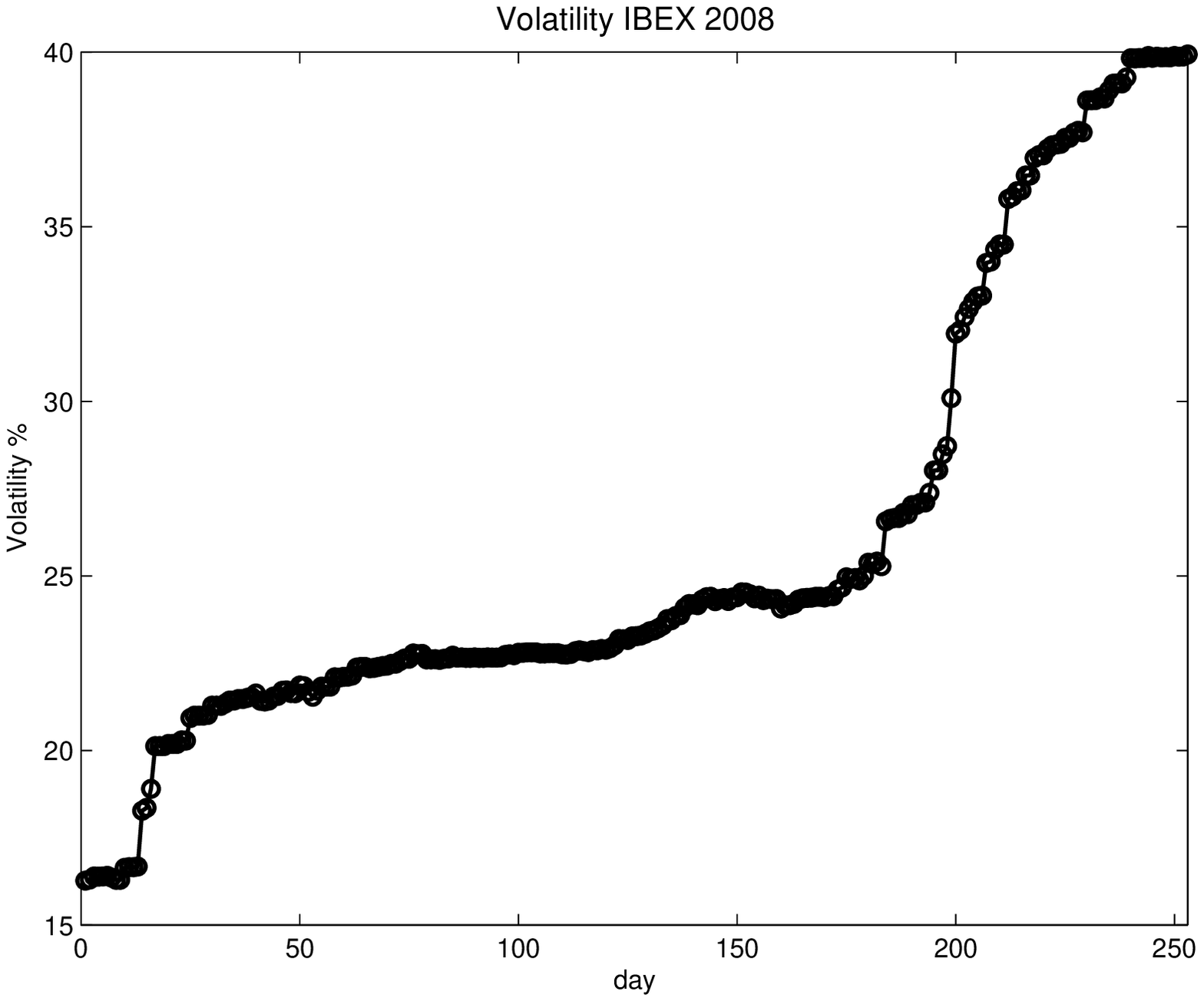}\\
  \includegraphics[width=0.48\textwidth]{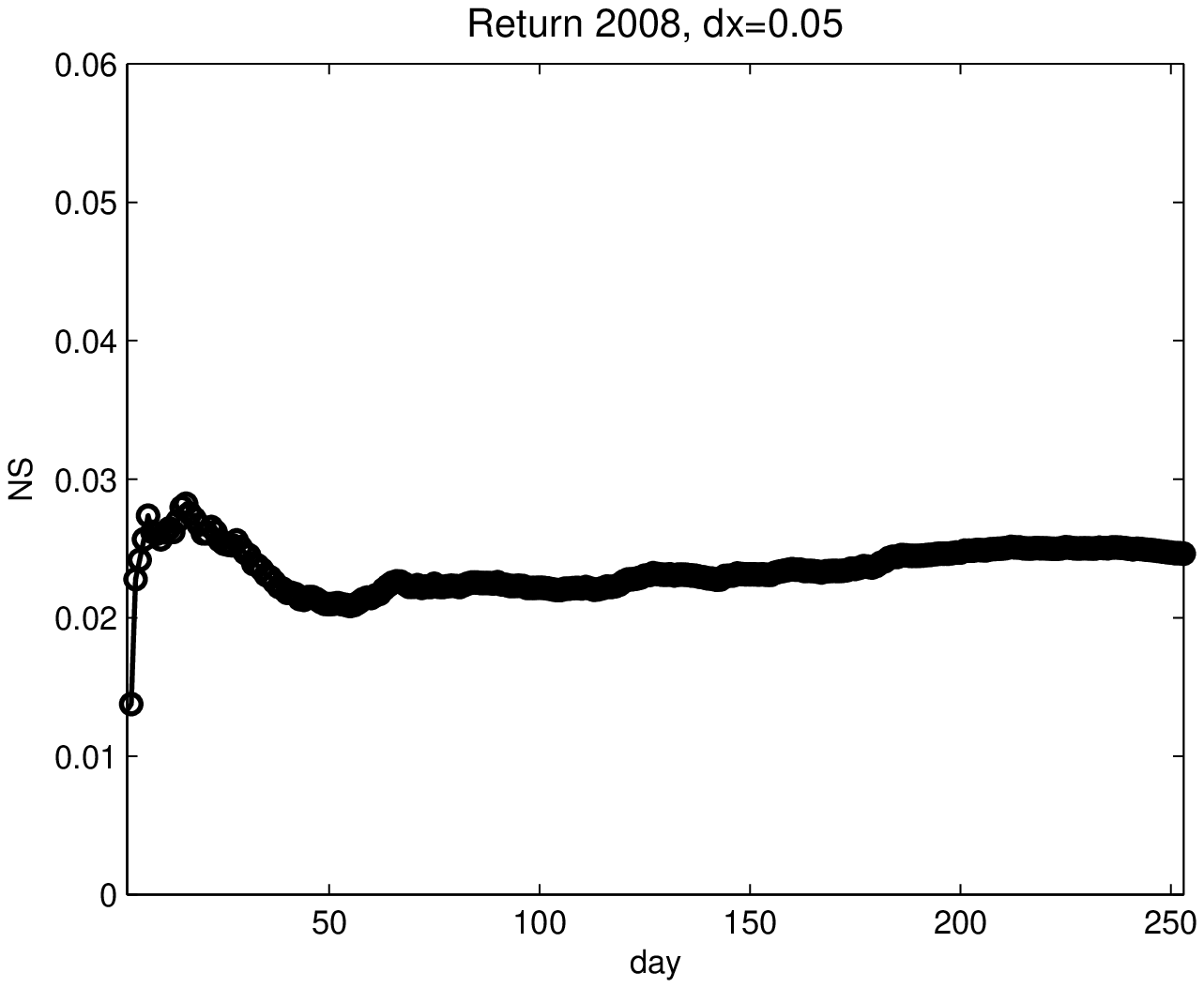}
  \includegraphics[width=0.48\textwidth]{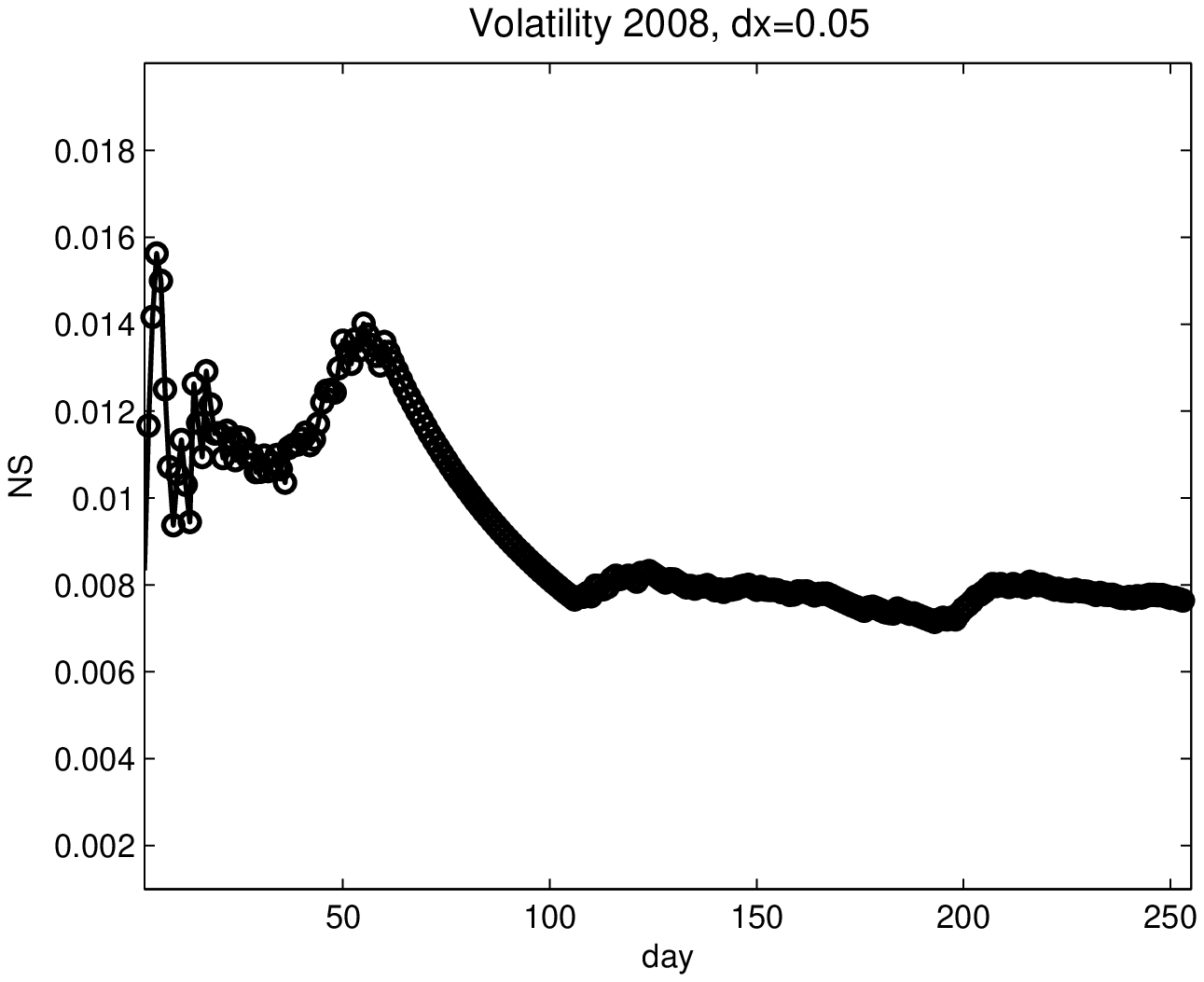}
\end{center}
\caption{Evolution of the Return of IBEX (panel (a)) Volatility of IBEX (panel (b)) Normalized Mean Strength ($NS$) for the competitivity of the return (panel (c)) and for the competitivity of the volatility (panel (d)) during 2008 and for $\Delta x=0.05$}
\label{fig2008}
\end{figure}

Finally, along 2013 the economical tensions in the international financial markets remitted, which improved the expectations about the global economic growth. The Spanish Stock Market followed these tendencies and gave positive results (with a return around $21,5\%$), outstripping a black period of three years of deep looses.  Figure~\ref{fig2013}, panel (a), shows the good evolution of return along 2013, that continued the tendency started in July 2012, while panel (b) illustrates the fall in volatility, specially clear in the fourth-quarter of the year.  The evolution of $NS$ for the competitivity graphs of return and volatility shows a soft increase along 2013 (see Figure~\ref{fig2013}, panels (c) and (d)) with values between 0.015 and 0.025 in the case of return and between 0.001 and 0.004 in the case of volatility. If we compare these values with the corresponding for 2004 and 2008, we can deduce that the rankings obtained for the return and volatility are more stable in 2013 than in the previous scenarios.

\begin{figure}[h!]
\begin{center}
  \includegraphics[width=0.48\textwidth]{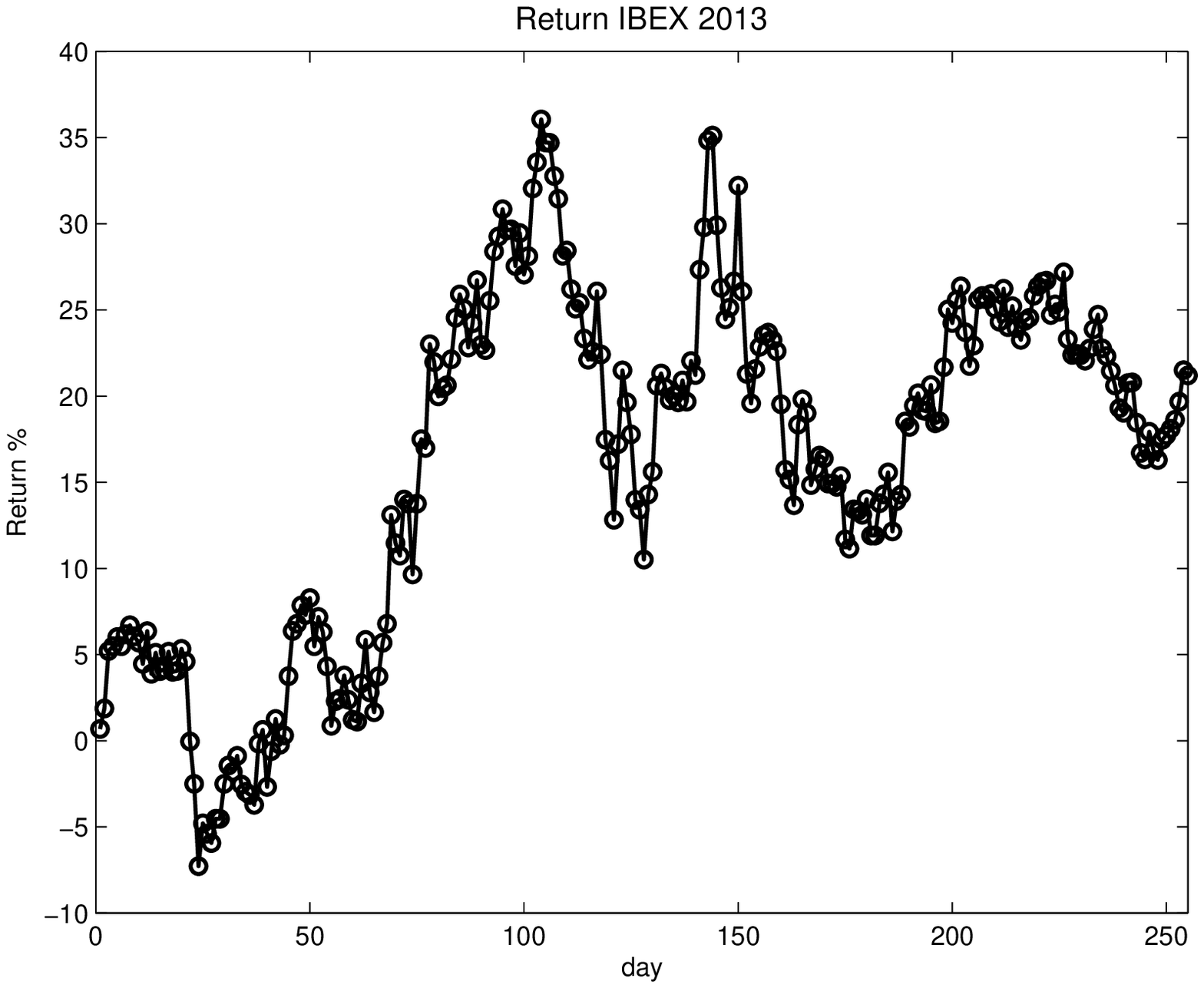}
  \includegraphics[width=0.48\textwidth]{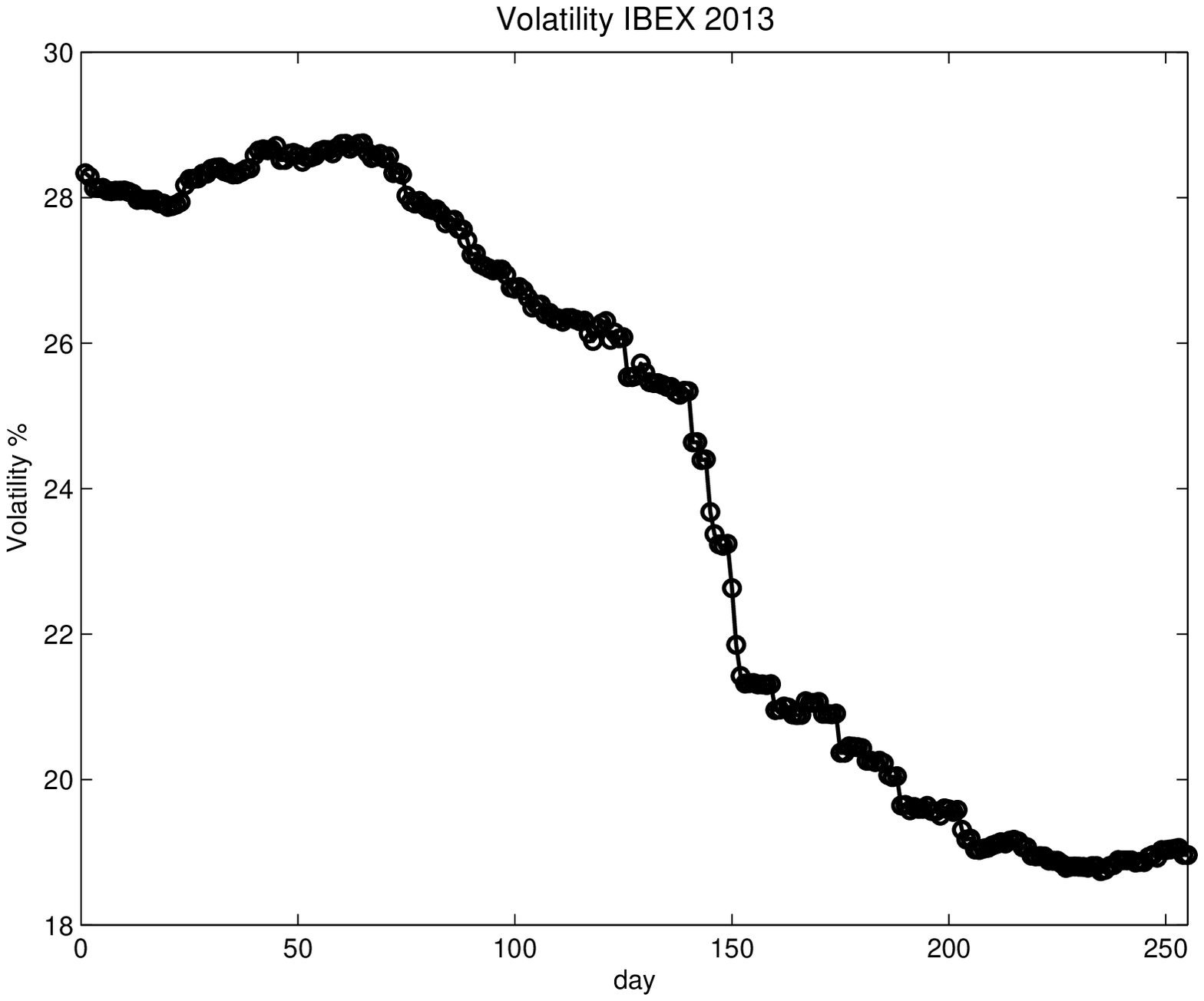}\\
  \includegraphics[width=0.48\textwidth]{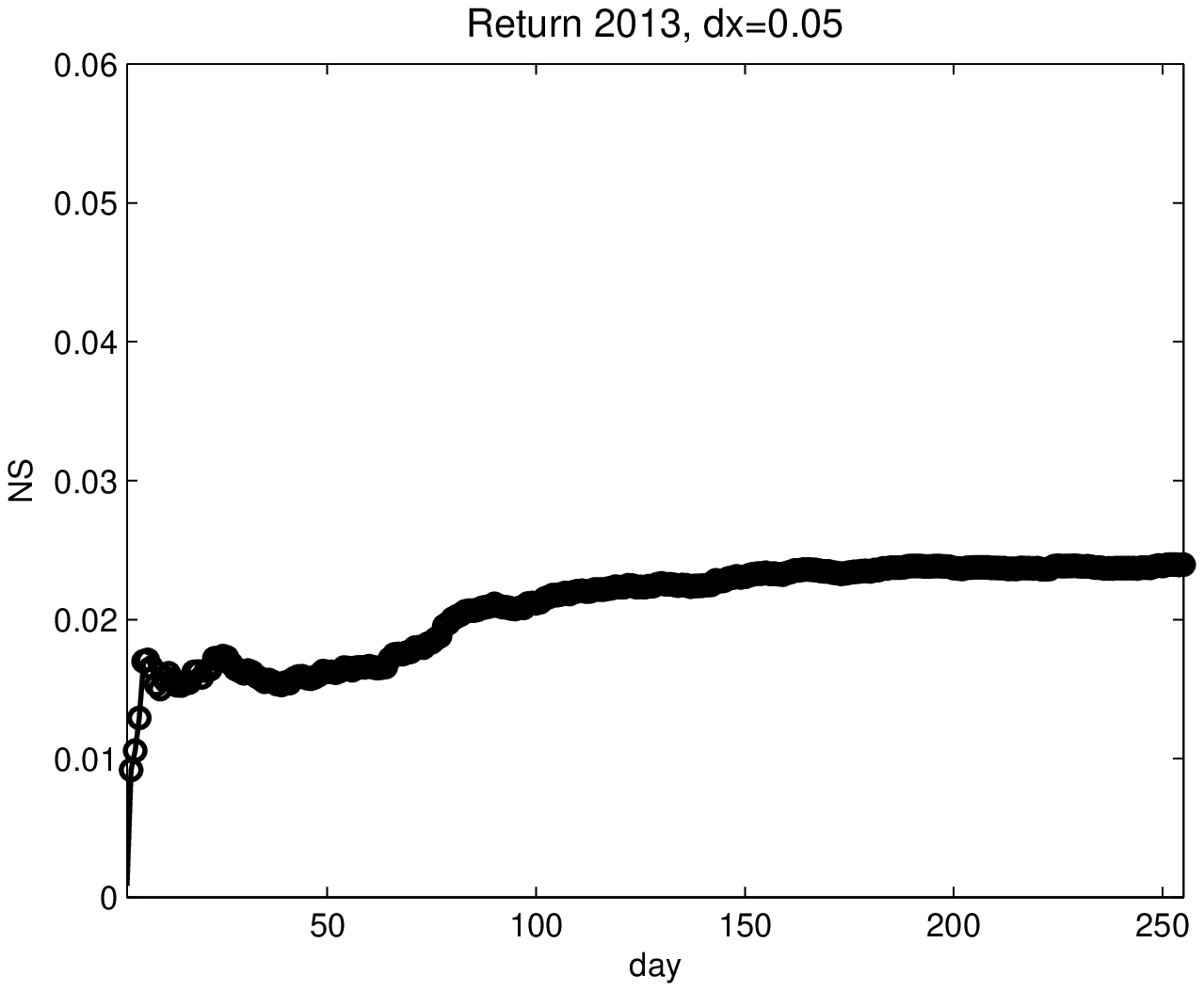}
  \includegraphics[width=0.48\textwidth]{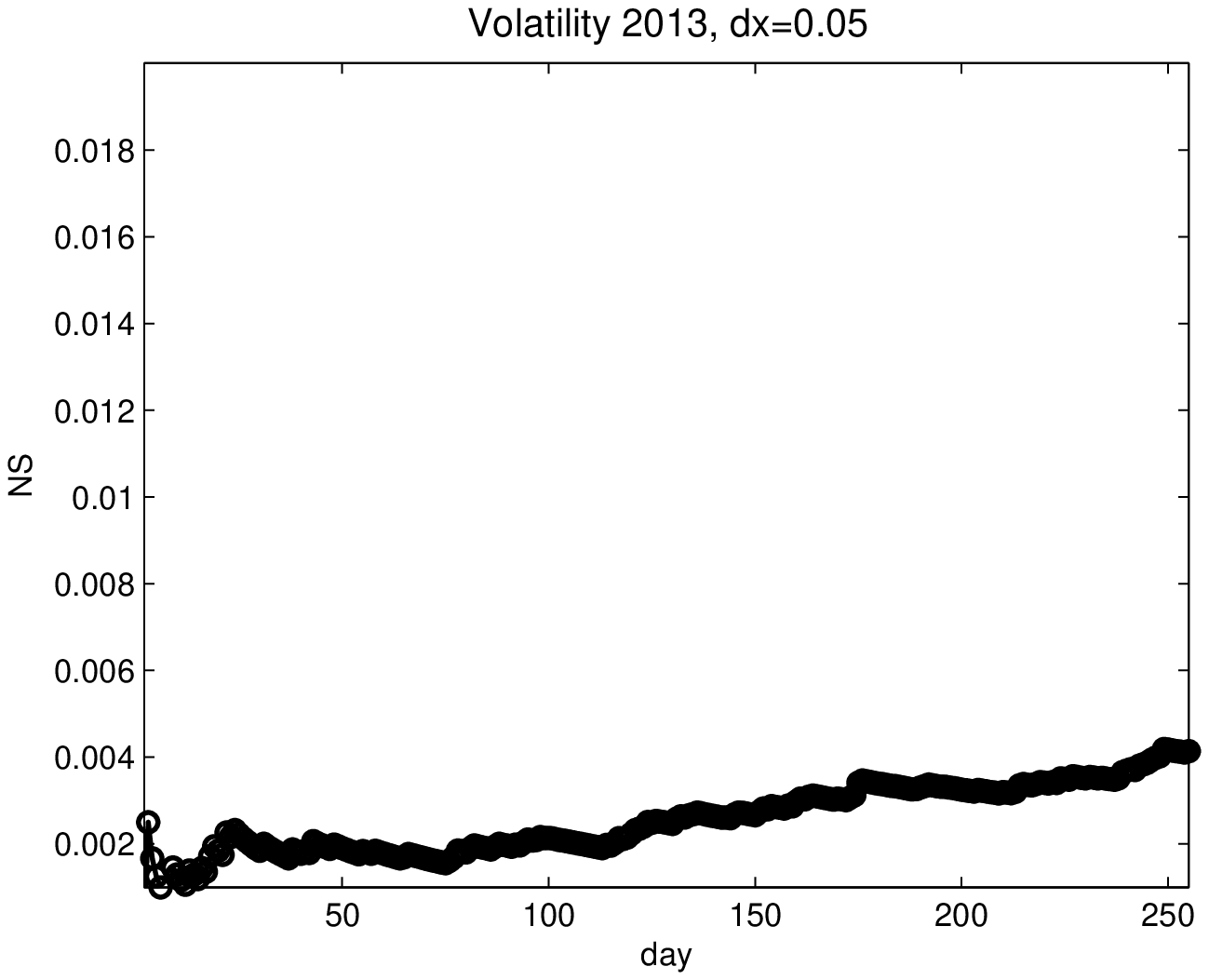}
\end{center}
\caption{Evolution of the Return of IBEX (panel (a)) Volatility of IBEX (panel (b)) Normalized Mean Strength ($NS$) for the competitivity of the return (panel (c)) and for the competitivity of the volatility (panel (d)) during 2013 and for $\Delta x=0.05$}
\label{fig2013}
\end{figure}

In any case, Figures~\ref{fig2004},~\ref{fig2008}~and~\ref{fig2013} shows that the information encapsulated in the competetivity graphs of the return and volatility is different from the information included in the classic analysis of return and volatility themselves and therefore it could be considered in order to give a sharper analysis of the stock markets. Further studies can be done by considering other structural parameters of the competitivity graph, including clustering or modularity, among others, and they would give information about the fluctuations in the rankings of the stock markets.

In addition to the analysis of the projected competitivity graph (including $NS$, degree distribution, clustering, modularity,...), the multiplex nature of the network it is also remarkable and further information can be obtained if we have a look at the structure and correlations between different layers. Figures~\ref{figRent2013}~and~\ref{figVola2013} illustrate these phenomena.

\begin{figure}[h!]
\begin{center}
  \includegraphics[width=0.24\textwidth]{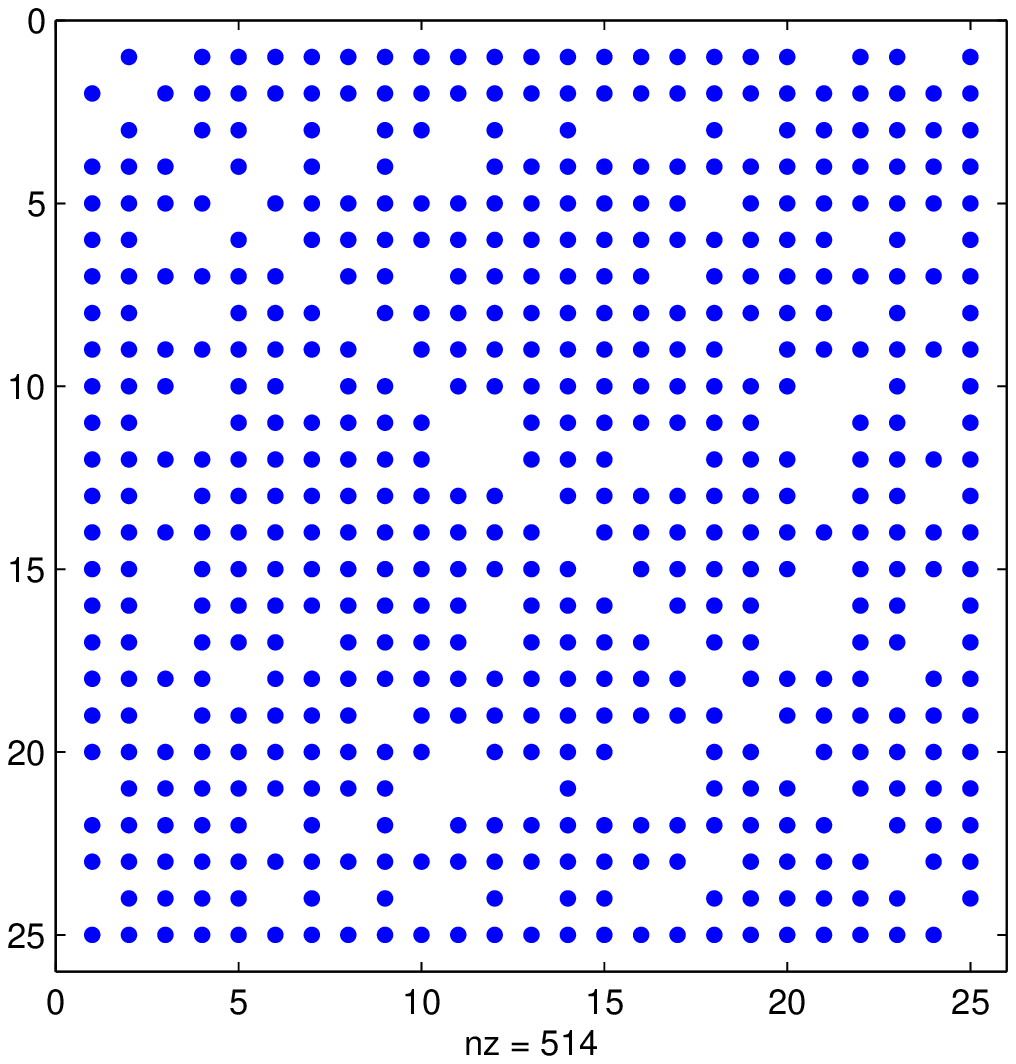}
  \includegraphics[width=0.24\textwidth]{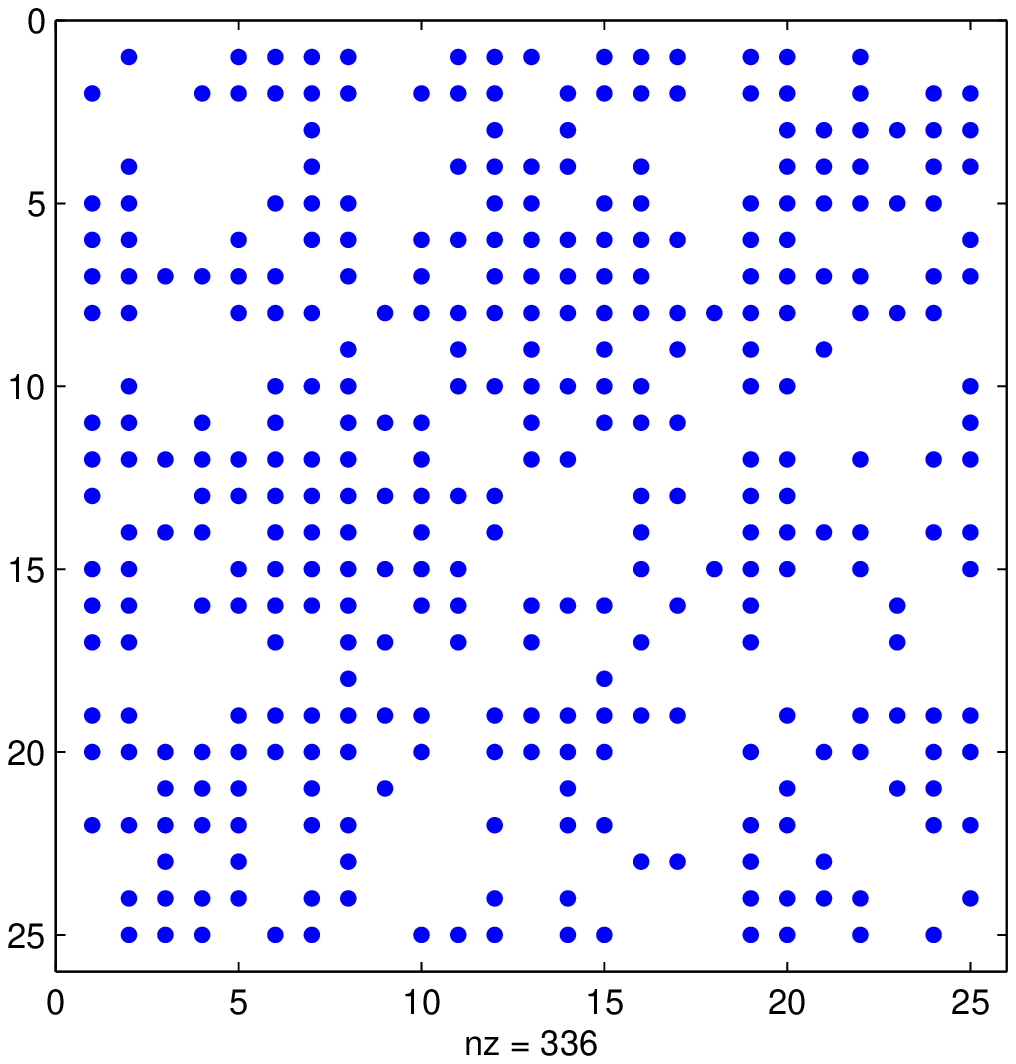}
  \includegraphics[width=0.24\textwidth]{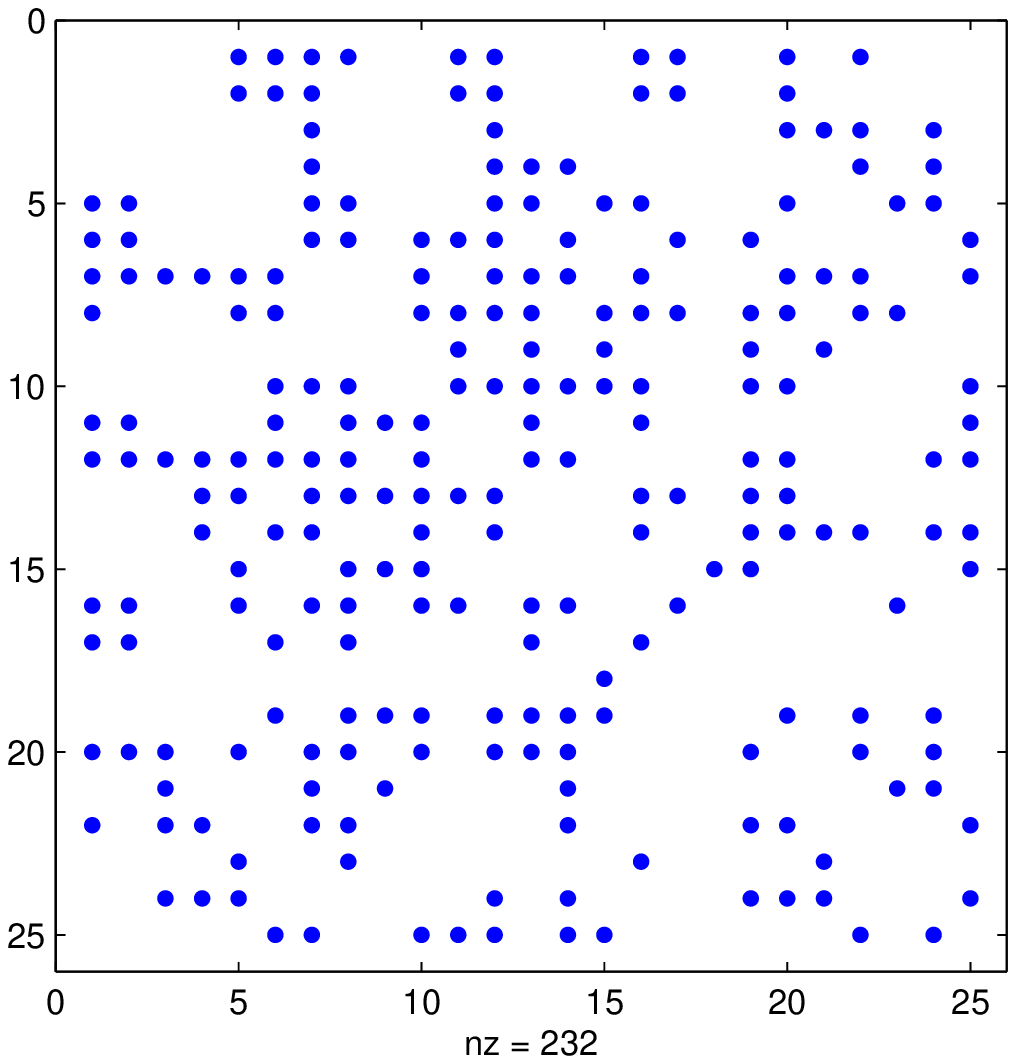}
  \includegraphics[width=0.24\textwidth]{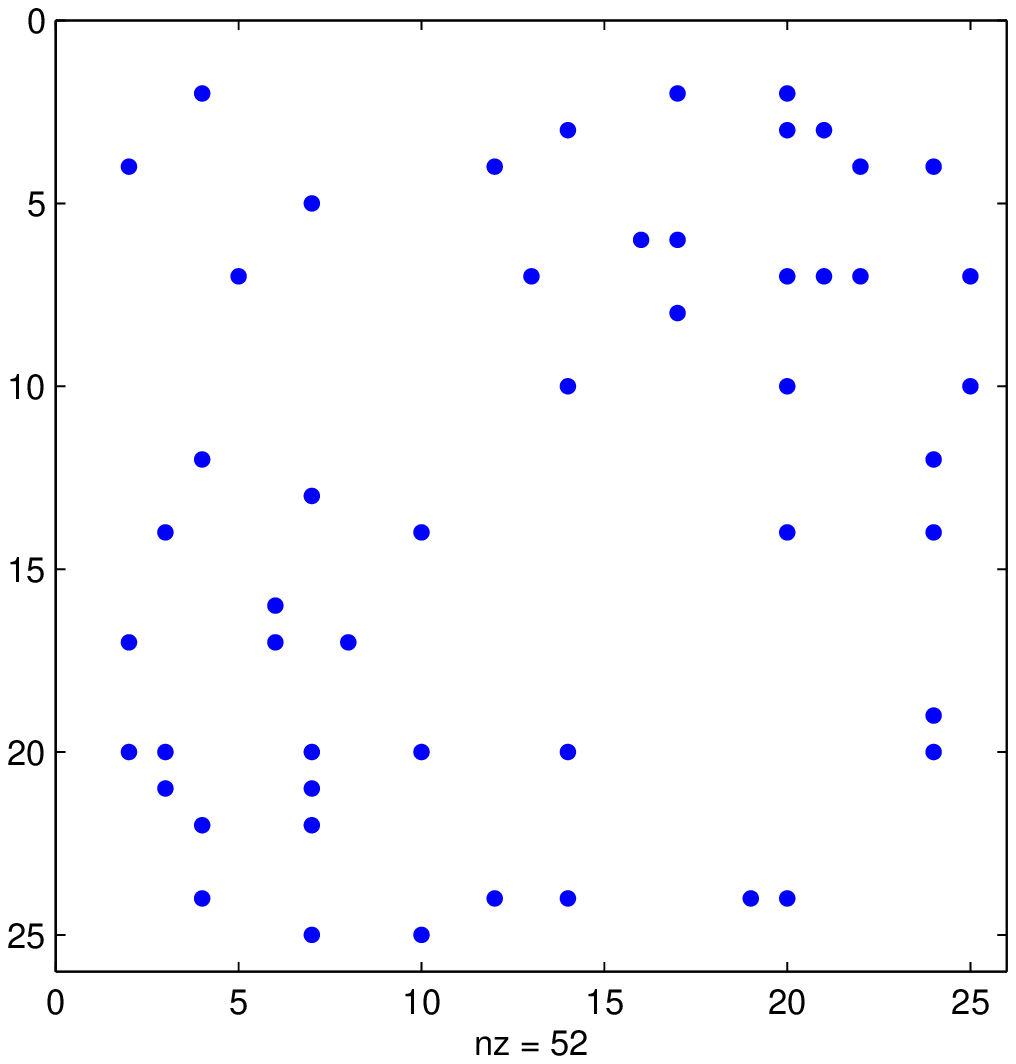}
\end{center}
\caption{The adjacency matrix of each layer of the multiplex competitivity graph of the return for the IBEX along 2013 with $\Delta x=0.005$, i.e. the crossing layer $\ell_1$ (panel (a)), the semi-crossing layer $\ell_2$ (panel (b)), the long-term-crossing layer $\ell_3$ (panel (c)) and the tie layer $\ell_4$ (panel (d))}
\label{figRent2013}
\end{figure}

Figure~\ref{figRent2013} presents the adjacency matrix of each layer of the competitivity graph obtained from the ranking of the return along 2013. Note that since the crossing layer $\ell_1$ (Figure~~\ref{figRent2013}, panel (a)) is quite dense, the number of crossings is very high, while the number of ties appearing in the tie layer $\ell_4$ (Figure~~\ref{figRent2013}, panel (d)) is quite low. The layer of long-term-crossing $\ell_3$ (panel (c)) it is always a subgraph of the semi-crossing layer $\ell_2$ (panel (b)), but in this case the two layers are not too similar.

\begin{figure}[h!]
\begin{center}
  \includegraphics[width=0.24\textwidth]{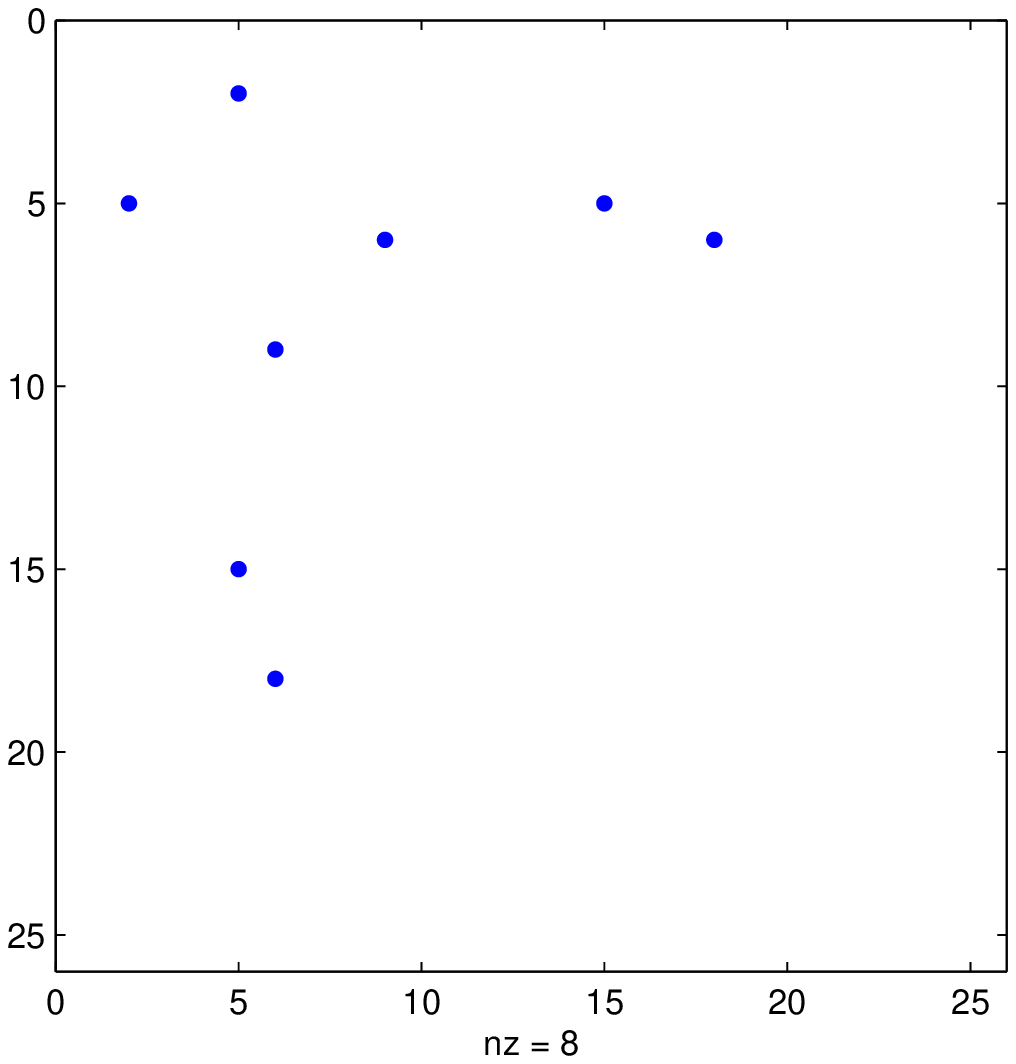}
  \includegraphics[width=0.24\textwidth]{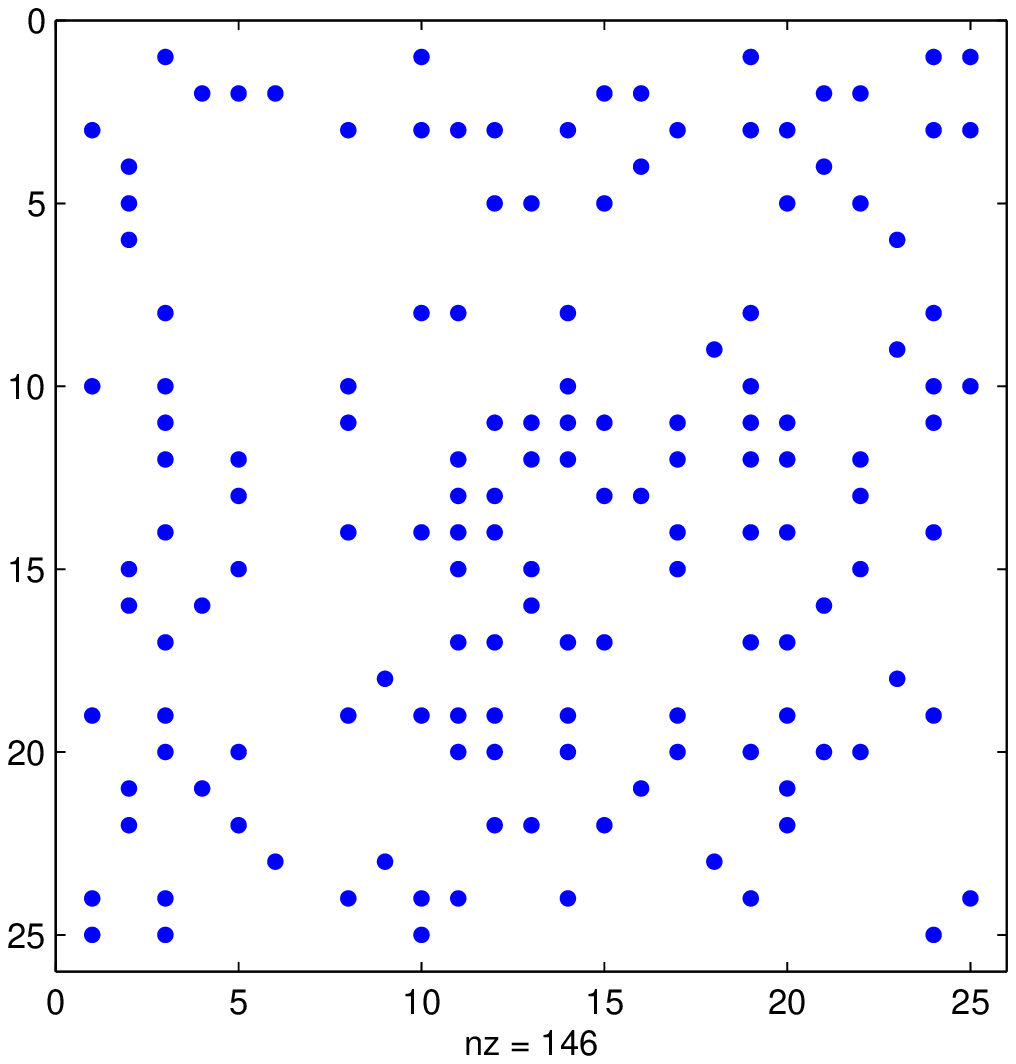}
  \includegraphics[width=0.24\textwidth]{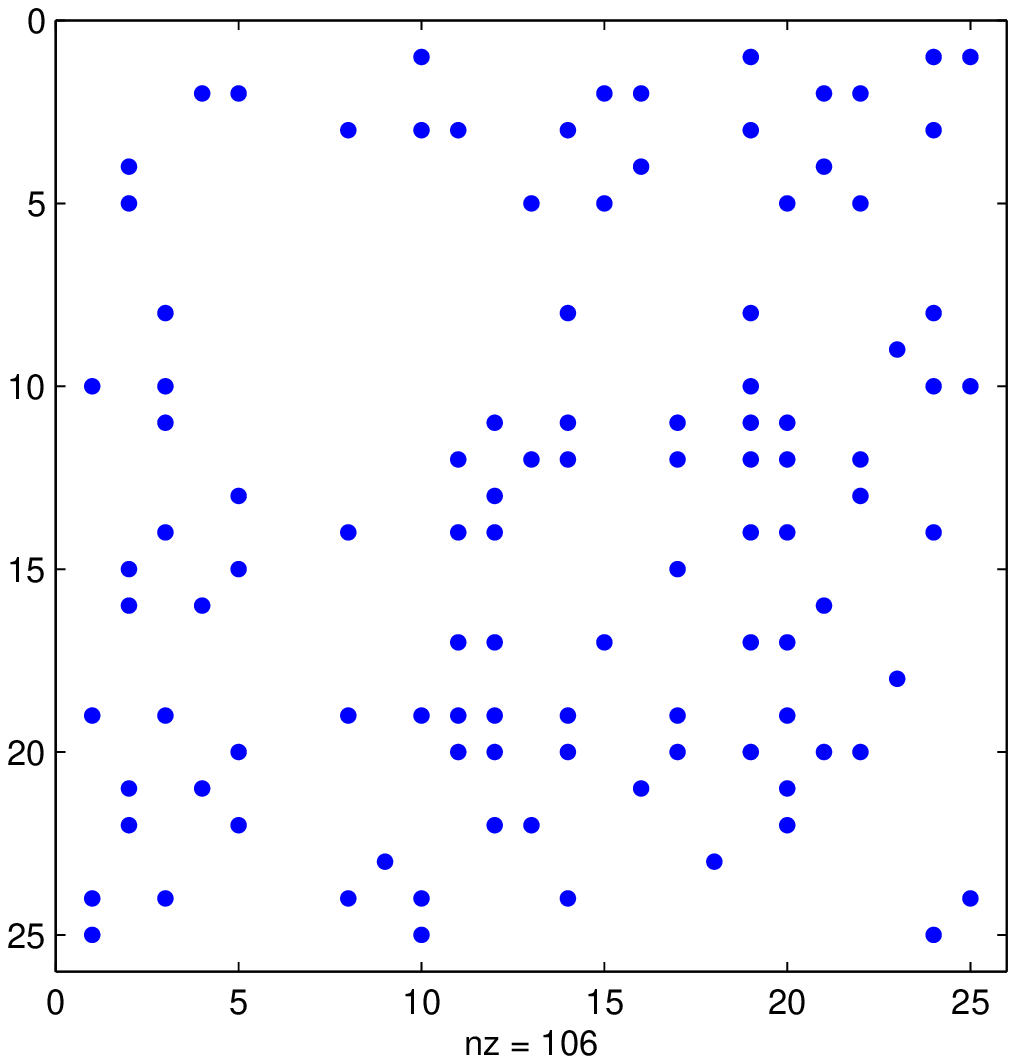}
  \includegraphics[width=0.24\textwidth]{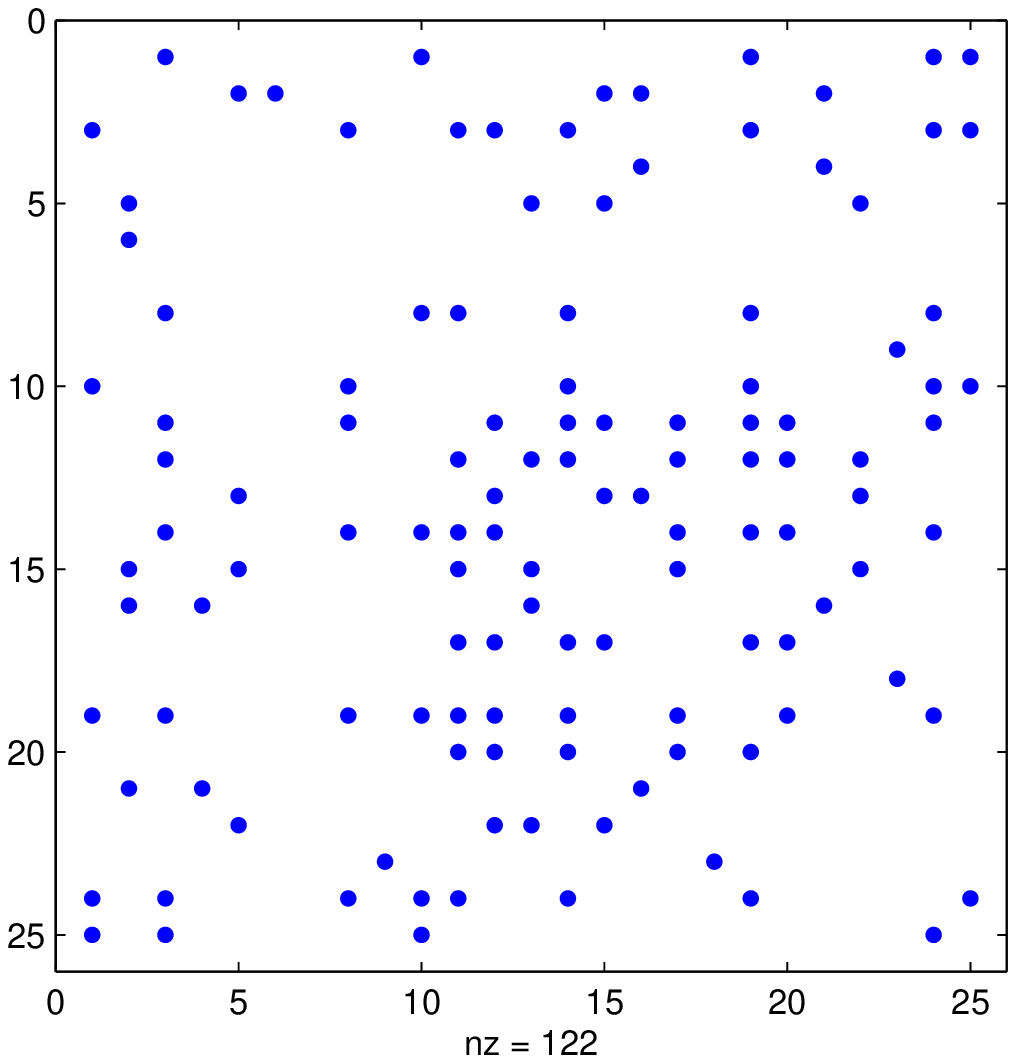}
\end{center}
\caption{The adjacency matrix of each layer of the multiplex competitivity graph of the volatility for the IBEX along 2013 with $\Delta x=0.005$, i.e. the crossing layer $\ell_1$ (panel (a)), the semi-crossing layer $\ell_2$ (panel (b)), the long-term-crossing layer $\ell_3$ (panel (c)) and the tie layer $\ell_4$ (panel (d))}
\label{figVola2013}
\end{figure}

Figure~\ref{figVola2013}  shows a similar analysis for the structure if the multiplex competitivity graph obtained from the ranking of the volatility along 2013, but the behaviour is different from the case of return for the same year (Figure~\ref{figRent2013}). In this case there is a very low number of crossings, since the crossing layer $\ell_1$ (Figure~~\ref{figVola2013}, panel (a)) is almost empty, while the number of ties is much higher than in the case of return. In addition to this there are stronger correlations among the semi-crossing layer $\ell_2$ (panel (b)), the layer of long-term-crossing $\ell_3$ (panel (c)) and the tie layer $\ell_4$ (panel (d)). Note that the multiplex behaviour of the multiplex competitivity graphs obtained from the ranking of the return and volatility are different, exhibiting these differences in multiplex properties.

\section{Conclusions and future work}

The classic idea of comparing rankings by using the number of crossings -or permutations- that occur when going from one ranking to the other, can be straightforward extended to compare series of several rankings when one focus on pairs of consecutive rankings. This work was already done in~\cite{Cri13} by introducing a tool from Complex Networks: the so-called {\sl Competitivity Graph}. In this paper, we have extended this procedure for the case when ties are allowed in each ranking. Given that the treatment of ties in rankings has a long tradition in the literature we have adopted (and extended) the procedures defined in~\cite{Fa06} to compute Kendall's distance between rankings when ties are presented. We have used the technique of multiplex networks and we have considered the crossings distinguishing when the crossing occurs between adjacent rankings (we take account of this on layer 1, or crossing layer), when the crossing is a change from tied to untied or vice versa (layer 2, or semi-crossing layer), when the crossing occurs after a period of consecutive ties (layer 3, or long-term crossing) and when there are a situation of tie among two elements in two consecutive rankings (layer 4, or tie layer) see Figure~\ref{figRent2013}.  We have denoted this network as the multiplex evolutive competitivity network associated to a family of rankings with ties.  We have also contributed with two theoretical issues. On the one hand, we have introduced a technique to convert a family of scores (or ratings) to a family of rankings by introducing the concept of approximated ties: that is, a way to define a tie when two elements have a score close enough depending on a certain threshold. On the other hand, we have shown theoretical results relating our concept of Normalized Mean Strength ($NS$) with the corresponding background concept of evolutive Kendall's distance. As an application of the introduced concepts we have shown an analysis of the Spanish stock market for 25 values during the period 2003-2013. The main conclusions of this analysis are the following:

\begin{itemize}
\item Since the number of crossings depends on the number -and the type- of ties presented in the rankings, the first step in an analysis of a series of scores is to study how the variations in the threshold for ties (that we denote $\Delta x$)  affects the number of ties (see Figure~\ref{figtiesall}). By using the data of return for year 2003 we have shown (see Figure~\ref{figFN1}) that our treatment of crossings is  equivalent to the model of~\cite{Fa06}  in the two limit cases: when no ties are present (that is, for small values of $\Delta x$) or when all the elements are tied (that is, when $\Delta x$ is very large) .

\item During the period 2003-2013 the maximum competitivity regarding the return values corresponded to year 2004 (see Figure~\ref{figEvolNSvsYears} and Figure~\ref{figEvolNSeachYears}, for $\Delta x=0.05$ or $0.005$, where there are not too many ties). That is, a year that corresponded to a growing period of the Spanish stock market, presented a lot of crossings between the values of the companies. In fact, such number of crossings ($NS=0.07$ at the end of the year, for $\Delta x=0.05$) has not reached yet, not even in the year 2013 that was also a growing period but with a level of crossings of $NS=0.042$. Therefore we conclude that the competitivity level (as measured by $NS$) of 2004 has not been reached yet.

\item We have shown that although the values of the return may show a great variation (for example see Figure~\ref{fig2004}) the competitivity may rest invariant or with minimal changes: that is, stock values may go up and down preserving the number of crossings between their values: they go up and down maintaining the level of competitivity of the whole system. That is why we say that our information gives a different information than the usual one.

\item Regarding volatility, that is a measure of  the dispersion of the stock along a year, we recall that a crossing in volatility means  that stocks more volatile than other become less volatile and vice versa. In year 2004 (Figure~\ref{fig2004}, down right) we see that the crossings in volatility tend to decrease, while the values of volatility decrease (Figure~\ref{fig2004} up right). In the year 2008 (a crisis year) volatility was growing up (Figure~\ref{fig2008} up right) while the crossings in volatility grew up from days around 40 to 60 in that year (reaching values of $NS$ around 0.014 such as in 2004) and then the crossings fall down to values of $NS$ around 0.008 such as the values of 2004.  Therefore we see that two years with different trend in volatility show the same trend in $NS$. Observing Figure~\ref{fig2013} down right, we see that in the year 2013 the volatility was decreasing while the $NS$ was slowly growing up but without reaching the values of the year 2004.
\end{itemize}
A comparison of the evolution of $NS$ of the competitivity graphs obtained for the return and volatility along years 2004, 2008 and 2013 is presented in Figure~\ref{fig:comparison}.

\begin{figure}[h!]
\begin{center}
  \includegraphics[width=0.48\textwidth]{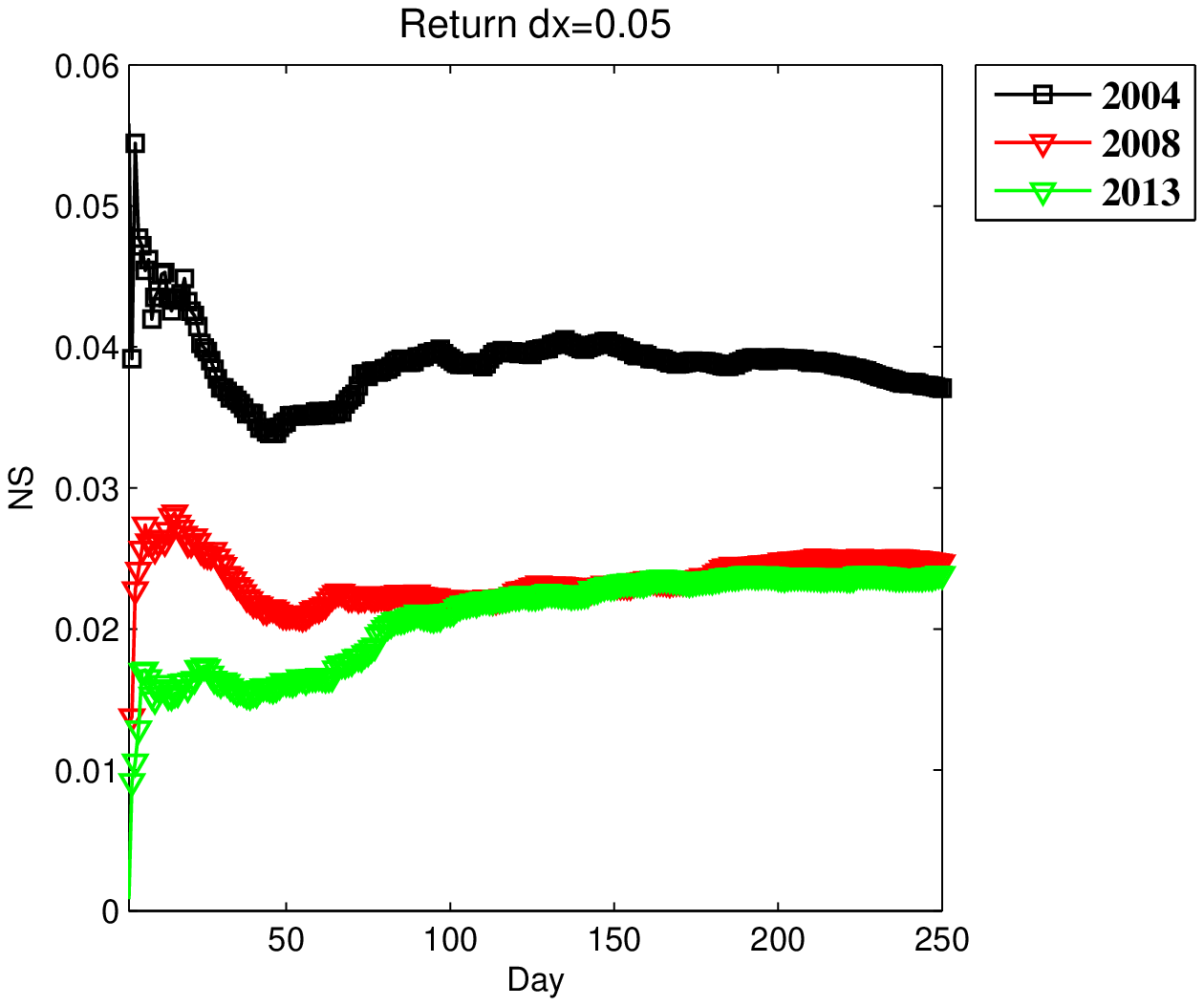}
  \includegraphics[width=0.48\textwidth]{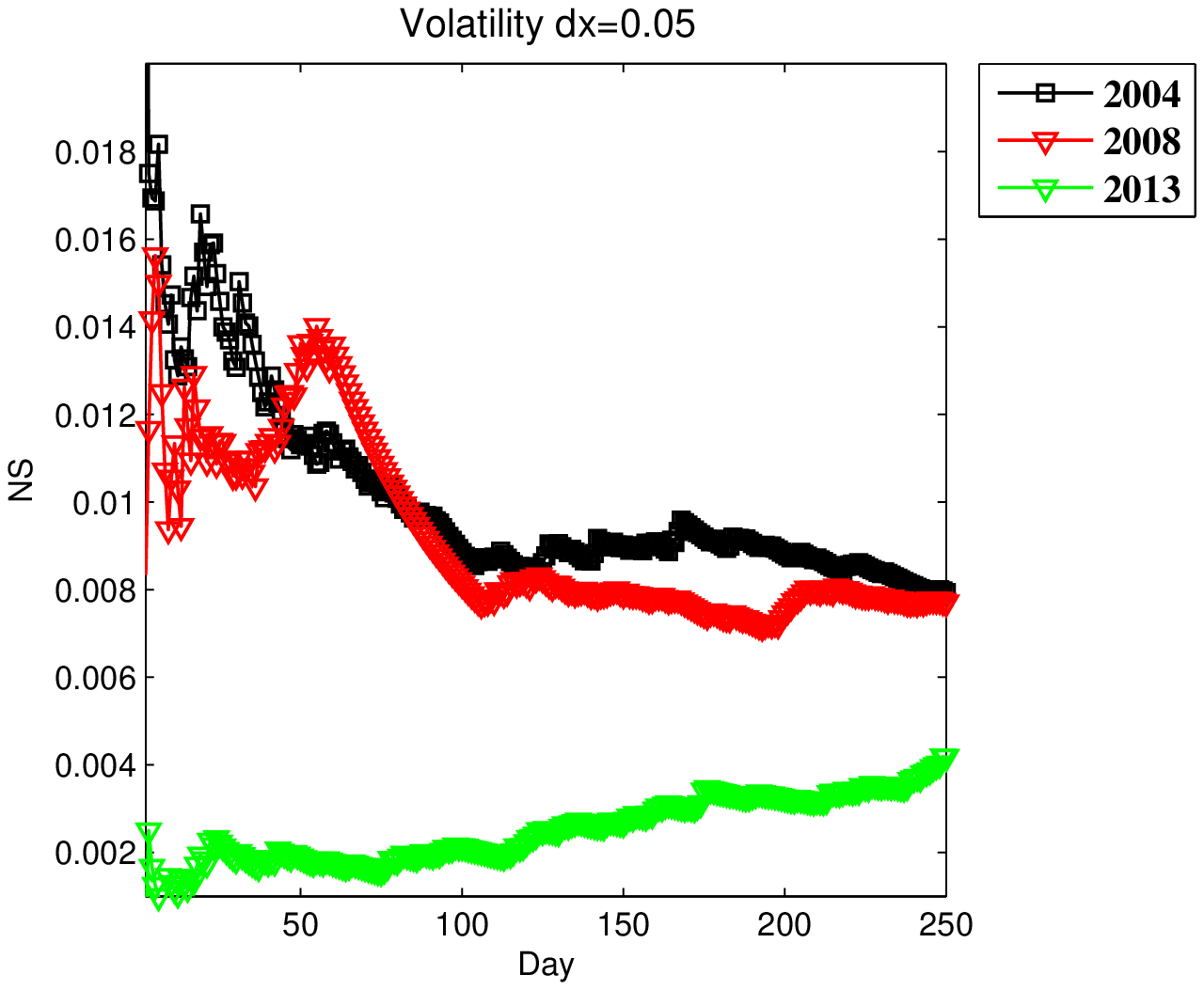}
\end{center}
\caption{A comparison of the Evolution of the Normalized Mean Strength ($NS$) for the competitivity of the return (panel (a)) and for the competitivity of the volatility (panel (b)) along years 2004 (in black), 2008 (in red) and 2013 (in green) with $\Delta x=0.05$}
\label{fig:comparison}
\end{figure}

\subsection{Future work}
Usual notions of (Classic) Complex Network Analysis and of Multiplex Complex Network Analysis, such as modularity, clustering coefficient, centrality, etc., could be studied for the (multiplex) networks associated to series of rankings with ties, and extract conclusions about the rankings themselves. The introduction of Multiplex Nature in the analysis of rankings with ties, enables the use of multiplex-native measures and techniques, such as inter-layer correlations, that could be considered in the future.

Let us illustrate this with an example: Imagine a group of cyclists going up to the top of a hill. They can go up in an ordered way, with no interchange in their relative positions or they can go up in a turbulent way changing their positions so that the head of the group is occupied by different cyclists as long as they climb the hill. Now imagine that each cyclist is a stock value and the group is composed of 25 values that are going up or down during a year. Our measure of competitiveness $NS$ gives an idea of the turbulences in this process of going up or down. If $NS$ is high it means that the group moves in a disordered way, crossing their positions. When $NS$ is low it means that the movement is made with small changes in the relative positions. We think it could be interesting to analyse if there are subgroups of stocks (or {\it communities} in the language of Complex Network Analysis) that move together in these crossings. It would be useful to analyse the crossings of each individual stock with respect to the others and could help analysing the optimal composition of an investment portfolio and studies about investment funds.

\section*{Acknowledgments}
This work was partially supported by Spanish MICINN Funds and FEDER Funds MTM2010-16153 and MTM2010-18674, and Junta de Andaluc\'{\i}a Funds FQM-264.


\end{document}